\documentclass[final,twocolumn]{IEEEtran}

\usepackage{amsmath,epsfig,amssymb,algorithm,algpseudocode,amsthm,cite,url}
\usepackage{hyperref}
\usepackage{tabularx}

\newtheorem{theorem}{Theorem}
\newtheorem{lemma}[theorem]{Lemma}
\newtheorem{remark}{Remark}
\newtheorem{assumption}{Assumption~A-\kern-0pt}

\newtheorem{corollary}[theorem]{Corollary}
\newcommand*\rfrac[2]{{}^{#1}\!\!/_{#2}}

      \DeclareMathOperator{\tr}{tr}
      \newcommand{\bh}{{\bf h}}
      \newcommand{\bH}{{\bf H}}
      
      \newcommand{\wbHh}{\widehat{\bf H}^{\mbox{\tiny H}}}
      
      \newcommand{\wbh}{\widehat{\bf h}}

      \newcommand{\C}{{\mathbb C}}
      \newcommand{\E}{{\mathbb E}}
      
      \newcommand{\bPhi}{\boldsymbol{\Phi}}
      \newcommand{\lb}{\left. }

      \newcommand{\rabs}{\right| }
      \newcommand{\bsigma}{\boldsymbol{\Sigma}}
      \newcommand{\maximize}[1]{{\underset{{#1}}{\mathrm{maximize}}}}

\usepackage{pgfplots}

\title{Linear Precoding Based on Polynomial Expansion: Large-Scale Multi-Cell MIMO Systems}

\author{Abla Kammoun,~\IEEEmembership{Member,~IEEE,}
        Axel M\"uller,~\IEEEmembership{Student Member,~IEEE,}
        Emil~Bj\"ornson,~\IEEEmembership{Member,~IEEE,}
        and~M\'erouane~Debbah,~\IEEEmembership{Senior~Member,~IEEE}
\thanks{Copyright (c) 2014 IEEE. Personal use of this material is permitted. However, permission to use this material for any other purposes must be obtained from the IEEE by sending a request to pubs-permissions@ieee.org}%
\thanks{A.~Kammoun, A.~M\"uller, E.~Bj\"ornson, and M.~Debbah are with the Alcatel-Lucent Chair on Flexible Radio, SUPELEC, Gif-sur-Yvette, France (e-mail: \{abla.kammoun, axel.mueller, emil.bjornson, merouane.debbah\}@supelec.fr). E.~Bj\"ornson is also with the Signal Processing Lab, ACCESS Linnaeus Centre, KTH Royal Institute of Technology, Stockholm, Sweden.}%
\thanks{E. Bj{\"o}rnson was with the Alcatel-Lucent Chair on Flexible Radio, Sup{\'e}lec,
Gif-sur-Yvette, France, and with the Department of Signal Processing, KTH
Royal Institute of Technology, Stockholm, Sweden. He is currently with the
Department of Electrical Engineering (ISY), Link{\"o}ping University, Link{\"o}ping, Sweden (email: emil.bjornson@liu.se)}%
\thanks{This research was funded by the International Postdoc Grant 2012-228 from The Swedish Research Council. It has been also supported by the ERC Starting Grant 305123 MORE (Advanced Mathematical Tools for Complex Network Engineering).}
}

\begin{document}
      \maketitle

      \begin{abstract}
Large-scale MIMO systems can yield a substantial improvements in spectral efficiency for future communication systems. Due to the finer spatial resolution and array gain achieved by a massive number of antennas at the base station, these systems have shown to be robust to inter-user interference and the use of linear precoding appears to be asymptotically optimal. However, from a practical point of view, most precoding schemes exhibit prohibitively high computational complexity as the system dimensions increase. For example, the near-optimal regularized zero forcing (RZF) precoding requires the inversion of a large matrix. To solve this issue, we propose in this paper to approximate the matrix inverse by a truncated polynomial expansion (TPE), where the polynomial coefficients are optimized to maximize the system performance. This technique has been recently applied in single cell scenarios and it was shown that a small number of coefficients is sufficient to reach performance similar to that of RZF, while it was not possible to surpass RZF.

In a realistic multi-cell scenario involving large-scale multi-user MIMO systems, the optimization of RZF precoding has, thus far, not been feasible. This is mainly attributed to the high complexity of the scenario and the non-linear impact of the necessary regularizing parameters. On the other hand, the scalar coefficients in TPE precoding give hope for possible throughput optimization. To this end, we exploit random matrix theory to derive a deterministic expression of the asymptotic signal-to-interference-and-noise ratio for each user based on channel statistics. We also provide an optimization algorithm to approximate the coefficients that maximize the network-wide weighted max-min fairness. The optimization weights can be used to mimic the user throughput distribution of RZF precoding. Using simulations, we compare the network throughput of the proposed TPE precoding with that of the suboptimal RZF scheme and show that our scheme can achieve higher throughput using a TPE order of only 5.

      \end{abstract}

      \section{Introduction}

A typical multi-cell communication system consists of $L>1$ base stations (BSs) that each are serving $K$ user terminals (UTs). The conventional way of mitigating inter-user interference in the downlink of such systems has been to assign orthogonal time/frequency resources to UTs within the cell and across neighboring cells. By deploying an array of $M$ antennas at each BSs, one can turn each cell into a multi-user multiple-input multiple-output (MIMO) system and enable flexible spatial interference mitigation \cite{Gesbert2010a}. The essence of downlink multi-user MIMO is \emph{precoding}, which means that the antenna arrays are used to direct each data signal spatially towards its intended receiver. The throughput of multi-cell multi-user MIMO systems ideally scales linearly with $\min(M,K)$. Unfortunately, the precoding design in multi-user MIMO requires very accurate instantaneous channel state information (CSI) \cite{Gesbert2007a} which can be cumbersome to achieve in practice \cite{Jindal2006a}. This is one of the reasons why only rudimentary multi-user MIMO techniques have found the way into current wireless standards, such as LTE-Advanced \cite{Holma2012a}.

Large-scale multi-user MIMO systems (with $M \gg K \gg 1$) have received massive attention lately \cite{Marzetta2010a,Rusek2013a,Hoydis2013a,WAG10}, partially because these systems are less vulnerable to inter-user interference. An exceptional spatial resolution is achieved when the number of antennas, $M$, is large; thus, the leakage of signal power caused by having imperfect CSI is less probable to arrive as interference at other users. Interestingly, the throughput of these systems become highly predictable in the large-($M,K$) regime; random matrix theory can provide simple deterministic approximations of the otherwise stochastic achievable rates \cite{HAC06,Evans08,WAG10,Muharar2011a,Hoydis2013a,COUbook}. These so-called \emph{deterministic equivalents} are tight as $M \rightarrow \infty$ due to channel hardening, but are often very accurate also at small/practical values of $M$ and $K$. The deterministic equivalents can, for example, be utilized for optimization of various system parameters \cite{WAG10}.

Many of the issues that made small-scale MIMO difficult to implement in practice appear to be solved by large-scale MIMO \cite{Rusek2013a}; for example, simple linear precoding schemes achieve (when $M\rightarrow \infty$ and $K$ is fixed) high performance in some multi-cell systems \cite{Rusek2013a} and robust to CSI imperfections \cite{Marzetta2010a}. The complexity of computing most of the state-of-the-art linear precoding schemes is, nevertheless, prohibitively high in the large-($M,K$) regime. For example, the optimal precoding parametrization in \cite{Bjornson2012c} and the near-optimal \emph{regularized zero-forcing (RZF)} precoding \cite{PEE05,WAG10,Hoydis2013a} require inversion of the Gram matrix of the joint channel of all users---this matrix operation has cubic complexity in $\min(M,K)$. A notable exception is the matched filter, also known as \emph{maximum ratio transmission (MRT)} \cite{Lo1999a}, which has only square complexity. This scheme is, however, not very appealing from a throughput perspective since it does not actively suppress inter-user interference and thus requires an order of magnitude more antennas to achieve performance close to that of RZF \cite{Hoydis2013a}.

In this paper, we propose to solve the precoding complexity issue by a new family of precoding schemes called truncated polynomial expansion (TPE) precoding. This family can be obtained by approximating the matrix inverse in RZF by a $(J-1)$-degree matrix polynomial which admits a low-complexity multistage hardware implementation. By changing $J$, one achieves a smooth transition in performance between MRT ($J=1$) and RZF ($J=\min(M,K)$). The hardware complexity of TPE precoding is proportional to $J$, thus the hardware complexity can be tailored to the deployment scenario. Furthermore, the TPE order $J$ needs not scale with the system dimensions $M$ and $K$ to maintain a fixed per-user rate gap to RZF, but it is desirable to increase it with the signal-to-noise ratio (SNR) and the quality of the CSI.

Building on the proof-of-concept provided by our work in \cite{Kammoun2014a} and the independent concurrent work of \cite{zarei}, this paper applies TPE precoding in a large-scale multi-cell scenario with realistic characteristics, such as user-specific channel covariance matrices, imperfect CSI, pilot contamination (due to pilot reuse in neighboring cells), and cell-specific power constraints. The $j$th BS serves its UTs using TPE precoding with an order $J_j$ that can be different between cells and thus tailored to factors such as cell size, performance requirements, and hardware resources.

In this paper, we derive new deterministic equivalents for the achievable user rates. The derivation of these expressions is the main analytical contribution and required major analytical advances related to the powers of stochastic Gram matrices with arbitrary covariances. The deterministic equivalents are tight when $M$ and $K$ grow large with a fixed ratio, but provide close approximations at small parameter values as well. Due to the inter-cell and intra-cell interference, the effective signal-to-interference-and-noise ratios (SINRs) are functions of the TPE coefficients in all cells. However, the deterministic equivalents only depend on the channel statistics, and not the instantaneous realizations, and can thus be optimized beforehand/offline. The joint optimization of all the polynomial coefficients is shown to be mathematically similar to the problem of multi-cast beamforming optimization considered in \cite{Sidiropoulos2006,Karipidis2008a,Gershman2010a}. We can therefore adapt the state-of-the-art optimization procedures from the multi-cast area and use these for offline optimization. We provide a simulation example that reveals that the optimized coefficients can provide even higher network throughput than RZF precoding at relatively low TPE orders, where TPE orders refers to the number of matrix polynomial terms.

\subsection{Notation}

Boldface (lower case) is used for column vectors, ${\bf x}$, and (upper case) for matrices, ${\bf X}$. Let ${\bf X}^{\mbox{\tiny T}}$, ${\bf X}^{\mbox{\tiny H}}$, and ${\bf X}^{*}$  denote the transpose, conjugate transpose, and conjugate of ${\bf X}$, respectively, while $\tr ({\bf X})$ denotes the matrix trace function. Moreover, $\mathbb{C}^{M\times K}$ denotes the set of matrices with size $M\times K$, whereas $\mathbb{C}^{M\times 1}$ is the set of vectors with size $M$. The $M\times M$ identity matrix is denoted by ${\bf I}_M$  and the ${\bf 0}_{M\times 1}$ stands for the $M\times 1$ vector with all entries equal to zero.
The expectation operator is denoted $\mathbb{E}[\cdot]$ and ${\rm var}[\cdot]$ denotes the variance. The spectral norm is denoted by $\|\cdot\|$ and equals the $L_2$ norm when applied to a vector. A circularly symmetric complex Gaussian random vector ${\bf x}$ is denoted ${\bf x} \sim \mathcal{CN}(\bar{{\bf x}},{\bf Q})$, where $\bar{{\bf x}}$ is the mean and ${\bf Q}$ is the covariance matrix.
For an infinitely differentiable monovariate function $f(t)$, the $\ell$th derivative at $t=t_0$ (i.e., $ \rfrac{d^\ell}{d t^\ell}f(t)|_{t=t_0} $) is denoted by $f^{(\ell)}(t_0)$ and more concisely by $f^{(\ell)}$ when $t=0$.
The big $\mathcal{O}$ notation $f(x) = \mathcal{O}(g(x))$ and little $o$ notation $f(x) = o(g(x))$  mean that
$\left| \frac{f(x)}{g(x)}\right|$ is bounded or approaches zero, respectively, as $x\rightarrow \infty$.

\section{System Model}

This section defines the multi-cell system with flat-fading channels, linear precoding, and channel estimation errors.

\subsection{Transmission Model}

We consider the downlink of a multi-cell system consisting of $L>1$ cells. Each cell consists of an $M$-antenna BS and $K$ single-antenna UTs. We consider a time-division duplex (TDD) protocol where the BS acquires instantaneous CSI in the uplink and uses it for the downlink transmission by exploiting channel reciprocity. We assume that the TDD protocols are synchronized across cells, such that pilot signaling and data transmission take place simultaneously in all cells.

The received complex baseband signal $y_{j,m}\in \mathbb{C}$ at the $m$th UT in the $j$th cell is
\begin{equation}
y_{j,m}=\sum_{\ell=1}^L {\bf h}_{\ell,j,m}^{\mbox{\tiny H}}{\bf x}_\ell + b_{j,m}
\label{eq:system_model}
\end{equation}
where ${\bf x}_\ell \in\mathbb{C}^{M \times 1}$ is the transmit signal from the $\ell$th BS and ${\bf h}_{\ell,j,m}\in\mathbb{C}^{M \times 1}$ is the channel vector from that BS to the $m$th UT  in the $j$th cell, and $b_{j,m}\sim\mathcal{CN}(0,\sigma^2)$ is additive white Gaussian noise (AWGN), with variance $\sigma^2$, at the receiver's input.

The small-scale channel fading is modeled as follows.

\begin{assumption}
The channel vector  ${\bf h}_{\ell,j,m}$ is modeled as
\begin{equation} \label{eq:channel-model}
{\bf h}_{\ell,j,m}={\bf R}_{\ell,j,m}^{\frac{1}{2}}{\bf z}_{\ell,j,m}
\end{equation}
where ${\bf z}_{\ell,j,m}\sim \mathcal{CN}({\bf 0}_{M\times 1},{\bf I}_M)$ and the channel covariance matrix ${\bf R}_{\ell,j,m} \in \mathbb{C}^{M \times M}$ satisfies the conditions
\begin{itemize}
	\item $\lim\sup_M \|{\bf R}_{\ell,j,m}\| <+\infty $, $\forall \ell,j,m$;
	\item $\lim\inf_M \frac{1}{M}\tr({\bf R}_{\ell,j,m}) >0$, $\forall \ell,j,m$.
	\end{itemize}
The channel vector has a fixed realization for a coherence interval and will then take a new independent realization.
This model is usually referred to as \emph{Rayleigh block-fading}.
\label{ass:channel}
\end{assumption}

The two technical conditions on ${\bf R}_{\ell,j,m}$ in Assumption A-\ref{ass:channel} enables asymptotic analysis and follow from the law of energy conservation and from increasing the physical size of the array with $M$; see \cite{Bjornson2014a} for a detailed discussion.

\begin{assumption}
All BSs use Gaussian codebooks and linear precoding. The precoding vector for the $m$th UT in the $j$th cell is ${\bf g}_{j,m} \in \C^{M \times 1}$ and its data symbol is $s_{j,m} \sim \mathcal{CN}(0,1)$.
\label{ass:linear-precoding}
\end{assumption}

Based on this assumption, the BS in the $j$th cell transmits the signal
\begin{equation}
{\bf x}_j = \sum_{m=1}^{K} {\bf g}_{j,m} s_{j,m}  ={\bf G}_j {\bf s}_j.
\label{eq:precoder_fun}
\end{equation}
The latter is obtained by letting ${\bf G}_j=\left[{\bf g}_{j,1},\ldots,{\bf g}_{j,K}\right] \in \C^{M \times K}$ be the precoding matrix of the $j$th BS and ${\bf s}_j = [s_{j,1} \, \ldots \, s_{j,K}]^{\mbox{\tiny T}} \sim \mathcal{CN}({\bf 0}_{K\times 1},{\bf I}_K) $ be the vector containing all the data symbols for UTs in the $j$th cell.
The transmission at BS $j$ is subject to a total transmit power constraint
\begin{equation}
\frac{1}{K}\tr  \left( {\bf G}_j {\bf G}_j^{\mbox{\tiny H}} \right) =P_j
\label{eq:power-constraint}
\end{equation}
where $P_j$ is the average transmit power per user in the $j$th cell.

The received signal \eqref{eq:system_model} can now be expressed as
\begin{equation}
y_{j,m}=\sum_{\ell=1}^L  \sum_{k=1}^{K}  {\bf h}_{\ell,j,m}^{\mbox{\tiny H}} {\bf g}_{\ell,k} s_{\ell,k}  + b_{j,m}.
\end{equation}
A well-known feature of large-scale MIMO systems is the channel hardening, which means that the effective useful channel $\bh_{j,j,m}^{\mbox{\tiny H}}{\bf g}_{j,m}$ of a UT converges to its average value when $M$ grows large. Hence, it is sufficient for each UT to have only statistical CSI and the performance loss vanishes as $M \rightarrow \infty$ \cite{Hoydis2013a}. An ergodic achievable information rate can be computed using a technique from \cite{Medard2000a}, which has been applied to large-scale MIMO systems in \cite{Marzetta2010a,jose,Hoydis2013a} (among many others).
The main idea is to decompose the received signal as
\begin{align*}
	{y}_{j,m}&=\mathbb{E}\left[\bh_{j,j,m}^{\mbox{\tiny H}}{\bf g}_{j,m}\right]{s}_{j,m}\\
		       &+\left(\bh_{j,j,m}^{\mbox{\tiny H}}{\bf g}_{j,m}-\mathbb{E}\left[\bh_{j,j,m}^{\mbox{\tiny H}}{\bf g}_{j,m}\right]\right) s_{j,m}\\
		       &+\sum_{(\ell,k) \neq (j,m)}\bh_{\ell,j,m}^{\mbox{\tiny H}}{\bf g}_{\ell,k}s_{\ell,k}+b_{j,m}
\end{align*}
and assume that the  channel gain $\mathbb{E}\left[ \left|\bh_{j,j,m}^{\mbox{\tiny H}}{\bf g}_{j,m} \right|^2 \right]$ is known at the corresponding UT, along with its variance ${\rm var}\left[{\bf h}_{j,j,m}^{\mbox{\tiny H}}{\bf g}_{j,m}\right] = \mathbb{E}\left[ \left| \bh_{j,j,m}^{\mbox{\tiny H}}{\bf g}_{j,m}-\mathbb{E}\left[\bh_{j,j,m}^{\mbox{\tiny H}}{\bf g}_{j,m}\right] \right|^2 \right]$ and the average sum interference power $\sum_{ (\ell,k) \neq (j,m) }\mathbb{E}[|{\bf h}_{\ell,j,m}^{\mbox{\tiny H}}{\bf g}_{\ell,k} |^2 ]$ caused by simultaneous transmissions to other UTs in the same and other cells.
By treating the inter-user interference (from the same and other cells) and channel uncertainty as worst-case Gaussian noise, UT $m$ in cell $j$ can  achieve the ergodic rate
$$r_{j,m} = \log_2 ( 1+ {\gamma}_{j,m})$$
without knowing the instantaneous values of ${\bf h}_{\ell,j,m}^{\mbox{\tiny H}} {\bf g}_{\ell, k}$ of its channel \cite{Medard2000a,Marzetta2010a,jose,Hoydis2013a}. The parameter $\gamma_{j,m}$ is given in \eqref{eq:gamma_jm} at the top of the next page and can be interpreted as the effective average SINR of the $m$th UT in the $j$th cell.

The last expression in \eqref{eq:gamma_jm} is obtained by using the following identities:
\begin{align*}
	{\rm var}({\bf h}_{j,j,m}^{\mbox{\tiny H}}{\bf g}_{j,m})&=\mathbb{E}\left[\left|{\bf h}_{j,j,m}^{\mbox{\tiny H}}{\bf g}_{j,m}\right|^2\right]\\
								&\quad -\left|\mathbb{E}\left[{\bf h}_{j,j,m}^{\mbox{\tiny H}}{\bf g}_{j,m}\right]\right|^2, \\
	\sum_{(\ell,k)\neq(j,m)} \mathbb{E}\left[\left|{\bf h}_{\ell,j,m}^{\mbox{\tiny H}}{\bf g}_{\ell,k}\right|^2\right]&=\sum_{\ell,k}\mathbb{E}\left[\left|{\bf h}_{\ell,j,m}^{\mbox{\tiny H}}{\bf g}_{\ell,k}\right|^2\right]\\				  &\quad -\mathbb{E}\left[\left|{\bf h}_{j,j,m}^{\mbox{\tiny H}}{\bf g}_{j,m}\right|^2\right].
\end{align*}

\begin{figure*}
\begin{align}
{\gamma}_{j,m}=\frac{\left|\E\left[{\bf h}_{j,j,m}^{\mbox{\tiny H}}{\bf g}_{j,m}\right]\right|^2}{\sigma^2+{\rm var}\left[{\bf h}_{j,j,m}^{\mbox{\tiny H}}{\bf g}_{j,m}\right]+\displaystyle{\sum_{
(\ell,k) \neq (j,m) } \mathbb{E}\left[\left|{\bf h}_{\ell,j,m}^{\mbox{\tiny H}}{\bf g}_{\ell,k}\right|^2\right]}} 
						       =\frac{\left|\E\left[{\bf h}_{j,j,m}^{\mbox{\tiny H}}{\bf g}_{j,m}\right]\right|^2}{\sigma^2+
\displaystyle{\sum_{\ell,k}} \, \mathbb{E}\left[\left|{\bf h}_{\ell,j,m}^{\mbox{\tiny H}}{\bf g}_{\ell,k}\right|^2\right]-\left|\mathbb{E}\left[{\bf h}_{j,j,m}^{\mbox{\tiny H}}{\bf g}_{j,m}\right]\right|^2}.
\label{eq:gamma_jm}
\end{align}
\hrulefill
\vspace*{2pt}
\end{figure*}

The achievable rates only depend on the statistics of the inner products ${\bf h}_{\ell,j,m}^{\mbox{\tiny H}}{\bf g}_{\ell,k}$ of the channel vectors and precoding vectors. The precoding vectors ${\bf g}_{j,m}$ should ideally be selected to achieve a strong signal gain and little inter-user and inter-cell interferences. This requires some instantaneous CSI at the BS, as described next.

\subsection{Model of Imperfect Channel State Information at BSs}

Based on the TDD protocol, uplink pilot transmissions are utilized to acquire instantaneous CSI at each BS.
Each UT in a cell transmits a mutually orthogonal pilot sequence, which allows its BS to estimate the channel to this user. Due to the limited channel coherence interval of fading channels, the same set of orthogonal sequences is reused in each cell; thus, the channel estimate is corrupted by pilot contamination emanating from neighboring cells \cite{Marzetta2010a}. When estimating the channel of UT $k$ in cell $j$, the corresponding BS takes its received pilot signal and correlates it with the pilot sequence of this UT. This results in the processed received signal
$$
{\bf y}_{j,k}^{\rm tr}={\bf h}_{j,j,k}+\sum_{\ell\neq j}{\bf h}_{j,\ell,k}+\frac{1}{\sqrt{\rho_{\rm tr}}}{\bf b}_{j,k}^{\rm tr}
$$
where ${\bf b}_{j,k}^{\rm tr}\sim \mathcal{CN}({\bf 0}_{M\times 1}, {\bf I}_M)$ and $\rho_{\rm tr}>0$ is the effective training SNR \cite{Hoydis2013a}.
The MMSE estimate $\wbh_{j,j,k}$ of $\bh_{j,j,k}$ is given as \cite{Bjornson2010a}:
\begin{align*}
\wbh_{j,j,k}&={\bf R}_{j,j,k}{\bf S}_{j,k}  {\bf y}_{j,k}^{\rm tr} \\
&={\bf R}_{j,j,k}{\bf S}_{j,k}\left(\sum_{\ell=1}^L {\bf h}_{j,\ell,k}+\frac{1}{\sqrt{\rho_{\rm tr}}}{\bf b}_{j,k}^{\rm tr}\right)
\end{align*}
where
$$
{\bf S}_{j,k}=\left(\frac{1}{\rho_{\rm tr}}{\bf I}_M+\sum_{\ell=1}^{L} {\bf R}_{j,\ell,k}\right)^{-1}  \quad \forall j,k
$$
and ${\bf R}_{j,j,k}$ is the channel covariance matrix of vector ${\bf h}_{j,j,k}$, as described in Assumption~A-\ref{ass:channel}.
The estimated channels from the $j$th BS to all UTs in its cell is denoted
\begin{equation}
\widehat{\bf H}_{j,j} = \left[ \wbh_{j,j,1} \, \ldots \, \wbh_{j,j,K} \right] \in \mathbb{C}^{M \times K}
\end{equation}
and will be used in the precoding schemes considered herein.

For notational convenience, we define the matrices
$$
\bPhi_{j,\ell,k}={\bf R}_{j,j,k}{\bf S}_{j,k}{\bf R}_{j,\ell,k}
$$
and note that $\wbh_{j,j,k}\sim\mathcal{CN}({\bf 0}_{M\times 1},\boldsymbol{\Phi}_{j,j,k})$ since the channels are Rayleigh fading and the MMSE estimator is used.



\section{Review on Regularized Zero-Forcing Precoding}
\label{sec:review-RZF}

The optimal linear precoding (in terms of maximal weighted sum rate or other criteria) is unknown under imperfect CSI and requires extensive optimization procedures under perfect CSI \cite{Bjornson2013d}. Therefore, only heuristic precoding schemes are feasible in fading multi-cell systems. Regularized zero-forcing (RZF) is a state-of-the-art heuristic scheme with a simple closed-form precoding expression \cite{PEE05,WAG10,Hoydis2013a}. The popularity of this scheme is easily seen from its many alternative names: transmit Wiener filter \cite{Joham2005a}, signal-to-leakage-and-noise ratio maximizing beamforming \cite{Sadek2007a}, generalized eigenvalue-based beamformer \cite{Stridh2006a}, and virtual SINR maximizing beamforming \cite{Bjornson2010c}. This section provides a brief review of prior performance results on RZF precoding in large-scale multi-cell MIMO systems. We also explain why RZF is computationally intractable to implement in practical large systems.

Based on the notation in \cite{Hoydis2013a}, the RZF precoding matrix used by the BS in the $j$th cell is
\begin{equation} \label{eq:RZF-definition}
{\bf G}_j^{\rm rzf}=\sqrt{K}\beta_j\left(\widehat{\bf H}_{j,j}\widehat{\bf H}_{j,j}^{\mbox{\tiny H}}+{\bf Z}_j+K\varphi_j{\bf I}_M\right)^{-1}\widehat{\bf H}_{j,j}
\end{equation}
where the scaling parameter $\beta_j$ is set so that the power constraint $\frac{1}{K}\tr \left( {\bf G}_j{\bf G}_j^{\mbox{\tiny H}} \right) =P_j$ in \eqref{eq:power-constraint} is fulfilled. The regularization parameters $\varphi_j$ and ${\bf Z}_j$ have the following properties.

\begin{assumption}
		\label{ass:regularization}
The regularizing parameter  $\varphi_j$ is strictly positive $\varphi_j>0$, for all  $j$.
The matrix ${\bf Z}_j$ is a deterministic Hermitian nonnegative definite matrix that satisfies $\lim\sup_{N} \frac{1}{N} \|{\bf Z}_j\|<+\infty$, for all $ j$.
\end{assumption}

Several prior works have considered the optimization of the parameter $\varphi_j$ in the single-cell case \cite{WAG10,Evans08} when ${\bf Z}_j={\bf 0}_{M\times M}$. This parameter provides a balance between maximizing the channel gain at each intended receiver (when $\varphi_j$ is large) and suppressing the inter-user interference (when $\varphi_j$ is small), thus $\varphi_j$  depends on the SNRs, channel uncertainty at the BSs, and the system dimensions \cite{PEE05,WAG10}. Similarly, the deterministic matrix ${\bf Z}_j$ describes a subspace where interference will be suppressed; for example, this can be the joint subspace spanned by (statistically) strong channel directions to users in neighboring cells, as proposed in \cite{Hosseini2013a}. The optimization of these two regularization parameters is a difficult problem in general multi-cell scenarios. To the authors' knowledge, previous works dealing with the multi-cell scenario have been restricted to considering intuitive choices of the regularizing parameters ${\varphi_j}$ and ${\bf Z}_j$. For example, this was recently done in \cite{Hoydis2013a}, where the performance of the RZF precoding was analyzed in the following asymptotic regime.


\begin{assumption}
	\label{ass:regime}
In the large-$(M,K)$ regime, $M$ and $K$ tend to infinity such that
	$$
	0< \liminf \frac{K}{M} \leq \limsup \frac{K}{M} < +\infty.
	$$	
\end{assumption}

In particular, it was shown in \cite{Hoydis2013a} that the SINRs perceived by the users tend to deterministic quantities in the large-$(M,K)$ regime. These quantities depend only on the statistics of the channels and are referred to as \emph{deterministic equivalents}.

In the sequel, by deterministic equivalent of a sequence of random variables $X_n$, we mean a deterministic sequence $\overline{X}_n$ which approximates $X_n$ such that
\begin{equation}
\E [X_n]-\overline{X}_n\xrightarrow[n\to+\infty]{}0.
\end{equation}

Before reviewing some results from \cite{Hoydis2013a}, we shall recall some  deterministic equivalents that play a key role in the next analysis. They are introduced in the following theorem.\footnote{We have chosen to work a slightly different definition of the deterministic equivalents than in \cite{Hoydis2013a}, since it fits better the analysis of our proposed precoding.}

\begin{theorem}[Theorem 1 in \cite{WAG10}]
\label{th:deterministic_1}
	Let ${\bf U}\in \mathbb{C}^{M\times M}$ have uniformly bounded spectral norm. Assume that matrix ${\bf Z}$ satisfies Assumption A-\ref{ass:regularization}. Let ${\bf H}\in\mathbb{C}^{M\times K}$ be a random matrix with independent column vectors ${\bf h}_j\sim \mathcal{CN}({\bf 0}_{M\times 1},{\bf R}_j)$ while the sequence of deterministic matrices ${\bf R}_j$  have uniformly bounded spectral norms. Denote by $\mathcal{R}$, the sequence of random matrices $\mathcal{R}=\left({\bf R}_k\right)_{k=1,\ldots,K}$ and by $\boldsymbol{\Sigma}(t)$ the resolvent matrix
	$$
	\boldsymbol{\Sigma}(t)=\left(\frac{t{\bf H}{\bf H}^{\mbox{\tiny H}}}{K}+\frac{t{\bf Z}}{K}+{\bf I}_M\right)^{-1}.
	$$
	Then, for any $t>0$ it holds that
	$$
	\frac{1}{K}\tr \left( {\bf U}\boldsymbol{\Sigma} \right) -\frac{1}{K}\tr \big( {\bf U}{\bf T}(t,\mathcal{R},{\bf Z})  \big) \xrightarrow[M,K\to+\infty]{\mathrm{a.s.}}0
	$$
	where ${\bf T}(t,\mathcal{R},{\bf Z})\in \mathbb{C}^{M\times M}$ is defined as
	$$
	{\bf T}(t,\mathcal{R},{\bf Z})=\left(\frac{1}{K}\sum_{k=1}^K \frac{t{\bf R}_{k}}{1+t\delta_k(t,\mathcal{R},{\bf Z})}+t\frac{1}{K}{\bf Z}+{\bf I}_M\right)^{-1}
	$$
	and the elements of $\boldsymbol{\delta}(t,\mathcal{R},{\bf Z})=\left[\delta_1(t,\mathcal{R},{\bf Z}),\ldots,\delta_K(t,\mathcal{R},{\bf Z})\right]^{\mbox{\tiny T}}$ are solutions to the following system of equations:
	\begin{align*}
		&\delta_{k}(t,\mathcal{R},{\bf Z})\\
						 &=\frac{1}{K}\tr \left( {\bf R}_k\left(\frac{1}{K}\sum_{j=1}^K\frac{t{\bf R}_j}{1+t\delta_j(t,\mathcal{R},{\bf Z})}+\frac{t}{K}{\bf Z}+{\bf I}_M\right)^{-1} \right).
\end{align*}
\end{theorem}
Theorem \ref{th:deterministic_1} shows how to approximate quantities with only one occurrence of the resolvent matrix $\bsigma(t)$. For many situations, this kind of result is sufficient to entirely characterize the asymptotic SINR, in  particular when dealing with the performance  of linear receivers  \cite{hoydis-asilomar,kammoun09}. However, when precoding is considered, random terms  involving two resolvent matrices arise, a case which is out of the scope of Theorem~\ref{th:deterministic_1}. For that, we recall the following result from \cite{WAG10}
which establishes deterministic equivalents for this kind of quantities.

\begin{theorem}[\cite{WAG10}]
\label{th:deterministic_2}
	Let $\boldsymbol{\Theta}\in \mathbb{C}^{M\times M}$ be Hermitian nonnegative definite with uniformly bounded spectral norm. Consider the setting of Theorem \ref{th:deterministic_1}. Then,
\begin{align*}
	&\frac{1}{K}\tr \left( {\bf U}\boldsymbol{\Sigma}(t)\boldsymbol{\Theta}\boldsymbol{\Sigma}(t) \right) - \frac{1}{K} \tr \left( {\bf U}\overline{\bf T}(t,\mathcal{R},{\bf Z},\boldsymbol{\Theta}) \right) \xrightarrow[M,K\to+\infty]{\mathrm{a.s.}} 0
\end{align*}
where
$$
\overline{\bf T}(t,\mathcal{R},{\bf Z},\boldsymbol{\Theta})={\bf T}\boldsymbol{\Theta}{\bf T}+t^2{\bf T}\frac{1}{K}\sum_{k=1}^K \frac{{\bf R}_k\overline{\delta}_k(t,\mathcal{R},{\bf Z},\boldsymbol{\Theta})}{(1+t\delta_k)^2}{\bf T},
$$
${\bf T}={\bf T}(t,\mathcal{R},{\bf Z})$, and $\boldsymbol{\delta}=\boldsymbol{\delta}(t,\mathcal{R},{\bf Z})$ are given by Theorem~\ref{th:deterministic_1}, and $\boldsymbol{\overline{\delta}}(t,\mathcal{R},{\bf Z},\boldsymbol{\Theta})=\left[\overline{\delta}_1(t,\mathcal{R},{\bf Z},\boldsymbol{\Theta}),\ldots,\overline{\delta}_K(t,\mathcal{R},{\bf Z},\boldsymbol{\Theta})\right]^{\mbox{\tiny T}}$ is computed as
$$
\boldsymbol{\overline{\delta}}=\left({\bf I}_K-t^2{\bf J}\right)^{-1}{\bf v}
$$
where ${\bf J}\in \mathbb{C}^{K\times K}$ and ${\bf v}\in\mathbb{C}^{K\times 1}$ are defined as
\begin{align*}
	\left[{\bf J}\right]_{k,\ell}&=\frac{\frac{1}{K}\tr \left( {\bf R}_k{\bf T}{\bf R}_\ell{\bf T} \right) }{K(1+t\delta_\ell)^2}, \ \  1\leq k,\ell\leq K \\
	\left[{\bf v}\right]_k&=\frac{1}{K} \tr \left( {\bf R}_k{\bf T}\boldsymbol{\Theta}{\bf T} \right), \ \ 1\leq k\leq K.
\end{align*}
\end{theorem}


\begin{remark}
	Note that the elements $\overline{\delta}_\ell$ are deterministic equivalents of $\frac{1}{K}\tr \left( {\bf R}_\ell \boldsymbol{\Sigma}(u)\boldsymbol{\Theta}\boldsymbol{\Sigma}(t) \right)$ in the sense that
	$$
	\frac{1}{K}\tr \left( {\bf R}_\ell \boldsymbol{\Sigma}(u)\boldsymbol{\Theta}\boldsymbol{\Sigma}(t) \right) - \overline{\delta}_\ell \xrightarrow[M,K\to+\infty]{\mathrm{a.s.}}0.
	$$
	Also, one can check that $\left(\overline{\delta}_k\right)_{k=1}^K$ is to $\overline{\bf T}$ as  $\left(\delta_k\right)_{k=1}^K$ is to ${\bf T}$, since
	$$
	\delta_k=\frac{1}{K}\tr \left( {\bf R}_k{\bf T}  \right) \ \ \textnormal{and}\ \ \overline{\delta}_k=\frac{1}{K} \tr \left( {\bf R}_k\overline{\bf T} \right).
	$$
\end{remark}
The performance of RZF precoding depends on a sequence of deterministic equivalents which we denote by $\left({\bf T}_\ell\right)_{\ell=1}^{L}$ and $\left(\overline{\bf T}_\ell\right)_{\ell=1}^L$. These are defined as
\begin{align*}
	{\bf T}_\ell&={\bf T}\left(\frac{1}{\varphi_\ell},\left(\boldsymbol{\Phi}_{\ell,\ell,k}\right)_{k=1}^K,{\bf Z}_\ell\right), \ \ell=1,\ldots,L \\
	\overline{\bf T}_\ell&=\overline{\bf T}\left(\frac{1}{\varphi_\ell},\left(\boldsymbol{\Phi}_{\ell,\ell,k}\right)_{k=1}^K,{\bf Z}_\ell,\frac{1}{\varphi_\ell}{\bf Z}_\ell+{\bf I}_M\right), \ \ell=1,\ldots,L.
\end{align*}
We are now in position to state the result establishing the convergence of the SINRs with RZF precoding. 
\begin{theorem}[Simplified from \cite{Hoydis2013a}]
	\label{th:theorem_rzf}
Denote by $\overline{\beta}_j$, $\theta_{\ell,j,m}$, $\kappa_{\ell,j,m}$, $\overline{\theta}_{\ell,j,m}$ and $\overline{\kappa}_{\ell,j,m}$ the deterministic quantities given by
\begin{align*}
	\overline{\beta}_j&=\frac{1}{\frac{1}{\varphi_j}\frac{1}{K}\tr ({\bf T}_j) - \frac{1}{K\varphi_j}\tr (\overline{\bf T}_j) } \\
	{\theta}_{\ell,j,m}&=\frac{1}{K}\tr ({\bf R}_{\ell,j,m}{\bf T}_\ell) \\
	\overline{\theta}_{\ell,j,m}&=\frac{1}{K}\tr ({\bf R}_{\ell,j,m}\overline{\bf T}_\ell) \\
	\kappa_{\ell,j,m}&=\frac{1}{K}\tr (\bPhi_{\ell,j,m}{\bf T}_\ell) \\
	\overline{\kappa}_{\ell,j,m}&=\frac{1}{K}\tr (\bPhi_{\ell,j,m}\overline{\bf T}_\ell) \\
	\zeta_{j,m}&=\frac{1}{\varphi_j+\delta_{j,m}}.
	\end{align*}
	The SINR at the $m$th user in the $j$th cell converges to $\overline{\gamma}_{j,m}$, where $\overline{\gamma}_{j,m}$ is given in \eqref{eq:gamma_jm_new} at the top of the next page.
	\begin{figure*}
		\begin{equation} \overline{\gamma}_{j,m}=\frac{\overline{\beta}_j(\delta_{j,m}\zeta_{j,m})^2}{\left(\displaystyle\sum_{\ell=1}^L\frac{\overline{\beta}_\ell}{\varphi_\ell}(\theta_{\ell,j,m}-\zeta_{\ell,m}\kappa_{\ell,j,m}^2)-\frac{\overline{\beta}_\ell}{\varphi_\ell}\overline{\theta}_{\ell,j,m}+\frac{2\overline{\beta}_\ell}{\varphi_\ell}\overline{\kappa}_{\ell,j,m}\kappa_{\ell,j,m}\zeta_{\ell,m}-\frac{\overline{\beta}_\ell}{\varphi_\ell}\kappa_{\ell,j,m}^2\overline{\delta}_{\ell,m}\zeta_{\ell,m}^2\right)-\overline{\beta}_j(\delta_{j,m}\zeta_{j,m})^2}.
	\label{eq:gamma_jm_new}
\end{equation}
\hrulefill
\vspace*{2pt}
\end{figure*}
\end{theorem}

\subsection{Complexity Issues of RZF Precoding}
\label{subsec:complexity-RZF}

The SINRs achieved by RZF precoding converge in the large-$(M,K)$ regime to the deterministic equivalents in Theorem \ref{th:theorem_rzf}. However, the precoding matrices are still random quantities that need to be recomputed at the same pace as the channel knowledge is updated. With the typical coherence time of a few milliseconds, we thus need to compute the large-dimensional matrix inverse in \eqref{eq:RZF-definition} hundreds of times per second. The number of arithmetic operations needed for matrix inversion scales cubically in the rank of the matrix, thus this matrix operation is intractable in large-scale systems; we refer to \cite{Kammoun2014a,Shepard2012a,zarei} for detailed complexity discussions. To reduce the implementation complexity and maintain most of the RZF performance, the low-complexity TPE precoding was proposed in \cite{Kammoun2014a} and \cite{zarei} for single-cell systems.
This new precoding scheme has two main benefits over RZF precoding: 1) the precoding matrix is not precomputed at the beginning of each coherence interval, thus there is no computational delays and the computational operations are spread out uniformly over time; 2) the precoding computation is divided into a number of simple matrix-vector multiplications which can be highly parallelized and can be implemented using a multitude of simple application-specific circuits. The next section extends this class of precoding schemes to practical multi-cell scenarios.

\section{Truncated Polynomial Expansion Precoding}

Building on the concept of truncated polynomial expansion (TPE), we now provide a new class of low-complexity linear precoding schemes for the multi-cell case. We recall that the TPE concept originates from the Cayley-Hamilton theorem which states that the inverse of a matrix ${\bf A}$ of dimension $M$ can be written as a weighted sum of its first $M$ powers:
\begin{equation*}
{\bf A}^{-1}=\frac{{(-1)}^{M-1}}{{\rm det}({\bf A})}\sum_{\ell=0}^{M-1}\alpha_\ell {\bf A}^{\ell}
\end{equation*}
where $\alpha_\ell$ are the coefficients of the characteristic polynomial.
A simplified precoding could, hence,  be obtained by taking only a truncated sum of the matrix powers. We refers to it as TPE precoding.

For ${\bf Z}_j={\bf 0}_{M\times M}$ and truncation order $J_j$, the proposed TPE precoding is given by the precoding matrix:
\begin{align} \label{eq_TPE-definition}
	{\bf G}_j^{\rm TPE}&=\sum_{n=0}^{J_j-1}w_{n,j} \left(\frac{\widehat{\bf H}_{j,j}\widehat{\bf H}_{j,j}^{\mbox{\tiny H}}}{K}\right)^{n}\frac{\widehat{\bf H}_{j,j}}{\sqrt{K}}\\	
			   &\triangleq \sum_{n=0}^{J_j-1}w_{n,j}{\bf V}_{n,j}\frac{\widehat{\bf H}_{j,j}}{\sqrt{K}} \notag
\end{align}
where
\begin{equation*}
{\bf V}_{n,j}=\left(\frac{\widehat{\bf H}_{j,j}\widehat{\bf H}_{j,j}^{\mbox{\tiny H}}}{K}\right)^n
\end{equation*}
and  $\left\{w_{n,j},j=0,\ldots,J_j-1\right\}$ are the $J_j$ scalar coefficients that are used in cell $j$. While RZF precoding
only has the design parameter $\varphi_j$, the proposed TPE precoding scheme offers a larger set of $J_j$ design parameters. These polynomial coefficients define a parameterized class of precoding schemes ranging from MRT (if $J_j=1$) to RZF precoding when $J_j=\min(M,K)$ and $w_{n,j}$ given by the coefficients based on the characteristic polynomial of $\sqrt{K}\left(\widehat{\bf H}_{j,j}\widehat{\bf H}_{j,j}+K\phi_j{\bf I}_M\right)^{-1}$. We refer to $J_j$ as the \emph{TPE order} corresponding to the $j$th cell and note that the corresponding polynomial degree in \eqref{eq_TPE-definition} is $J_j-1$. For any $J_j  < \min(M,K)$, the polynomial coefficients have to be treated as design parameters that should be selected to maximize some appropriate system performance metric \cite{Kammoun2014a}. An initial choice is
\begin{equation} \label{eq:coeff-initial}
w_{n,j}^{\text{initial}} =  \beta_j \kappa_j \sum_{m=n}^{J_j-1} \binom{m}{n} (1-\kappa_j \varphi_j)^{m-n} (-\kappa_j)^{n}
\end{equation}
where $\beta_j$ and $\varphi_j$ are as in RZF precoding, while the parameter $\kappa_j$ can take any value such that $\Big\|  {\bf I}_M - \kappa_j \Big(\frac{1}{K}\widehat{\bH}\wbHh +\varphi_j {\bf I}_M \Big)  \Big\| < 1$. This expression is obtained by calculating a Taylor expansion of the matrix inverse. The coefficients in \eqref{eq:coeff-initial} gives performance close to that of RZF precoding when $J_j \rightarrow \infty$ \cite{Kammoun2014a}. However, the optimization of the RZF precoding has not, thus far, been feasible.
Therefore, we can obtain even better performance than the suboptimal RZF, using only small TPE orders (e.g., $J_j=4$), if the coefficients are optimized with the system performance metric in mind. This optimization of the polynomial coefficients in multi-cell systems is dealt with in Subsection~\ref{sec:optimum} and the results are evaluated in Section~\ref{sec:simulations}.

A fundamental property of TPE is that $J_j$ needs not scale with the $M$ and $K$, because ${\bf A}^{-1}$ is equivalent to inverting each eigenvalue of ${\bf A}$ and the polynomial expansion effectively approximates each eigenvalue inversion by a Taylor expansion with $J_j$ terms \cite{Sessler2005a}. More precisely, this means that the approximation error per UT is only a function of $J_j$ (and not the system dimensions), which was proved for multiuser detection in \cite{Honig2001a} and validated numerically in \cite{Kammoun2014a} for TPE precoding.

\begin{remark}
	The deterministic matrix ${\bf Z}_j$ was used in RZF precoding to suppress interference in certain subspaces. Although the TPE precoding in \eqref{eq_TPE-definition} was derived for the special case of ${\bf Z}_j={\bf 0}_{M\times M}$, the analysis can easily be extended for arbitrary ${\bf Z}_j$. To show this, we define the rotated channels $\tilde{\bf h}_{\ell,j,m} = (\frac{{\bf Z}_j}{K}+\varphi_j{\bf I}_M )^{-1/2} {\bf h}_{\ell,j,m} \sim \mathcal{CN}({\bf 0}_{M\times 1}, (\frac{{\bf Z}_j}{K}+\varphi_j{\bf I}_M )^{-1/2} {\bf R}_{\ell,j,m} (\frac{{\bf Z}_j}{K}+\varphi_j{\bf I}_M )^{-1/2})$.
RZF precoding can now be rewritten as
\begin{equation} \label{eq:RZF-definition-rewritten}
{\bf G}_j^{\rm rzf}=\frac{\beta_j}{\sqrt{K}} \left( \frac{{\bf Z}_j}{K}+\varphi_j{\bf I}_M \right)^{\!-1/2} \left(\frac{\widehat{\tilde{\bf H}}_{j,j}\widehat{\tilde{\bf H}}_{j,j}^{\mbox{\tiny H}}}{K} + {\bf I}_M \right)^{\! -1} \widehat{\tilde{\bf H}}_{j,j}
\end{equation}
where $\widehat{\tilde{\bf H}}_{j,j} = (\frac{{\bf Z}_j}{K}+\varphi_j{\bf I}_M )^{-1/2} [ \hat{\bf h}_{j,j,1} \, \ldots \, \hat{\bf h}_{j,j,K}]$. When this precoding matrix is multiplied with a channel as ${\bf h}_{j,\ell,m}^{\mbox{\tiny H}} {\bf G}_j^{\rm rzf}$, the factor $(\frac{{\bf Z}_j}{K}+\varphi_j{\bf I}_M )^{-1/2}$ will also transform ${\bf h}_{j,\ell,m}$ into a rotated channel. By considering the rotated channels instead of the original ones, we can apply the whole framework of TPE precoding. The only thing to keep in mind is that the power constraints might be different in the SINR optimization of Section \ref{sec:optimum}, but the extension in straightforward.
\end{remark}

Next, we provide an asymptotic analysis of the SINR for TPE precoding.

\subsection{Large-Scale Approximations of the SINRs}

In this section, we show that in the large-($M,K$) regime, defined by Assumption A-\ref{ass:regime}, the  SINR experienced by  the $m$th UT served by the $j$th cell, can be approximated by a deterministic term, depending solely on the channel statistics. Before stating our main result, we shall cast \eqref{eq:gamma_jm} in a simpler form by introducing some extra notation.


Let ${\bf w}_j=\left[w_{0,j},\ldots,w_{J_j-1,j}\right]^{\mbox{\tiny T}}$ and let ${\bf a}_{j,m}\in \mathbb{C}^{J_j\times 1}$ and ${\bf B}_{\ell,j,m}\in \mathbb{C}^{J_j\times J_j}$ be given by
$$
\left[{\bf a}_{j,m}\right]_n=\frac{{\bf h}_{j,j,m}^{\mbox{\tiny H}}}{\sqrt{K}}{\bf V}_{n,j}\frac{\widehat{{\bf h}}_{j,j,m}}{\sqrt{K}}, \ \ n\in \left[0,J_j-1\right],
$$
$$
\left[{\bf B}_{\ell,j,m}\right]_{n,p}=\frac{1}{K} {\bf h}_{\ell,j,m}^{\mbox{\tiny H}}{\bf V}_{n+p+1,\ell}{\bf h}_{\ell,j,m}, \ \ n,p \in\left[0,J_\ell-1\right].
$$
Then, the SINR experienced by the $m$th user in the $j$th cell is
\begin{equation} \label{eq:SINR-based-on-TPE}
\gamma_{j,m}=\frac{\left|\mathbb{E}[ {\bf w}_j^{\mbox{\tiny T}}{\bf a}_{j,m}] \right|^2}{\frac{\sigma^2}{K}+\displaystyle{\sum_{\ell=1}^{L}}\mathbb{E}\left[{\bf w}_\ell^{\mbox{\tiny T}}{\bf B}_{\ell,j,m}{\bf w}_\ell\right]-\left|\mathbb{E}[ {\bf w}_j^{\mbox{\tiny T}}{\bf a}_{j,m}]\right|^2}.
\end{equation}
Since ${\bf a}_{j,m}$ and ${\bf B}_{\ell,j,m}$ are of finite dimensions, it suffices to determine  an asymptotic approximation of the expected value of each of their elements. For that, similarly to our work in \cite{Kammoun2014a}, we link their elements to  the resolvent matrix
$$
\bsigma(t,j)=\left(t\frac{\widehat{\bf H}_{j,j}\widehat{\bf H}_{j,j}^{\mbox{\tiny H}}}{K}+{\bf I}_M\right)^{-1}
$$
by introducing the functionals ${ X}_{j,m}(t)$ and $Z_{\ell,j,m}(t)$
\begin{align} \label{eq:X-def}
	{ X}_{j,m}(t)&=\frac{1}{K}{\bf h}_{j,j,m}^{\mbox{\tiny H}}\bsigma(t,j)\widehat{\bf h}_{j,j,m} \\
	{ Z}_{\ell,j,m}(t)&=\frac{1}{K}{\bf h}_{\ell,j,m}^{\mbox{\tiny H}}\bsigma(t,\ell){\bf h}_{\ell,j,m} \label{eq:Z-def}
\end{align}
it is straightforward to see that:
\begin{align}
	\left[{\bf a}_{j,m}\right]_{n}&=\frac{(-1)^n}{n!}X_{j,m}^{(n)} \label{eq:a_jm} \\ 
	\left[{\bf B}_{\ell,j,m}\right]_{n,p}&=\frac{(-1)^{(n+p+1)}}{(n+p+1)!}Z_{\ell,j,m}^{(n+p+1)}\label{eq:B_ljm}
\end{align}
where $X_{j,m}^{(k)}\triangleq\lb\frac{d^{k}X_{j,m}(t)}{dt^k}\rabs_{\scriptscriptstyle t=0}$ and $Z_{\ell,j,m}^{(k)}\triangleq\left[{\lb\frac{d^{k}Z_{\ell,j,m}(t)}{d t^{k}}\rabs_{\scriptscriptstyle t=0}}\right]$.
Higher order moments of the spectral distribution of $\frac{1}{K}\widehat{\bf H}_{j,j}\widehat{\bf H}_{j,j}^{\mbox{\tiny H}}$ appear when taking derivatives of $X_{j,m}(t)$ or $Z_{\ell,j,m}(t)$.
The asymptotic convergence of these moments require an extra assumption ensuring that the spectral norm of $\frac{1}{K}\widehat{\bf H}_{j,j}\widehat{\bf H}_{j,j}^{\mbox{\tiny H}}$ is almost surely bounded. This assumption is expressed as follows.
\begin{assumption}
	The correlation matrices ${\bf R}_{\ell,j,m}$ belong to a finite-dimensional matrix space. This means that it exists a finite integer $S>0$ and a linear independent family of matrices ${\bf F}_1,\ldots,{\bf F}_S$ such that
	$$
	{\bf R}_{\ell,j,m}=\sum_{k=1}^S \alpha_{\ell,j,m,k}{\bf F}_k
	$$
	where $\alpha_{\ell,j,m,1},\ldots, \alpha_{\ell,j,m,S}$ denote the coordinates of ${\bf R}_{\ell,j,m}$ in the basis ${\bf F}_1,\ldots,{\bf F}_S$.
	\label{ass:correlation}
\end{assumption}
\begin{remark}
Two remarks are in order.
	\begin{enumerate}
		\item This condition is less restrictive than the one used in \cite{hoydis}, where  ${\bf R}_{\ell,j,m}$ is assumed to belong to a finite set of matrices.

\item Note that Assumption A-\ref{ass:correlation} is in agreement with several physical channel models presented in the literature. Among them, we distinguish the following models:
	\begin{itemize}
		\item The channel model of \cite{marzetta-icassp}, which considers a fixed number of dimensions or angular bins $S$ by letting
$$
{\bf R}_{\ell,j,m}^{\frac{1}{2}}=d_{\ell,j,m}^{-\frac{\theta}{2}}\left[{\bf K}, \ \ {\bf 0}_{M,M-S}\right]
$$
for some positive definite ${\bf K}\in \mathbb{C}^{M\times M-S}$, where $\theta$ is the path-loss exponent and $d_{\ell,j,m}$ is the distance between the $m$th user in the $j$th cell  and the $\ell$th cell.
\item The one-ring channel model with user groups from \cite{caire-12}. This channel model considers a finite number of groups ($G$ groups) which share approximately the same location and thus the same covariance matrix. Let $\theta_{\ell,j,g}$ and $\Delta_{\ell,j,g}$ be respectively the azimuth angle and the azimuth angular spread between the cell $\ell$ and the users in group $g$ of cell $j$. Moreover, let $d$ be the distance between two consecutive antennas (see Fig.~1 in \cite{caire-12}). Then, the $(u,v)$th entry of the covariance matrix ${\bf R}_{\ell,j,m}$ for users is group $g$ is
	 \begin{align}
		 &\left[{\bf R}_{\ell,j,m}\right]_{u,v}=\frac{1}{2\Delta_{\ell,j,g}}\int_{-\Delta_{\ell,j,g}+\theta_{\ell,j,g}}^{\Delta_{\ell,j,g}+\theta_{\ell,j,g}} e^{\jmath d(u-v)\sin\alpha}d\alpha \label{eq:correlation}\\
	  & (\textnormal{user } m \textnormal{ is in group }  g \textnormal{ of cell } j) .\nonumber
\end{align}
	
\end{itemize}

\end{enumerate}
\label{remark:correlation}
\end{remark}
Before stating our main result, we shall define (in a similar way, as in the previous section) the deterministic equivalents that will be used:
\begin{align*}
	{\bf T}_\ell(t)&={\bf T}\left(t,\left(\boldsymbol{\Phi}_{\ell,\ell,k}\right)_{k=1}^K,{\bf 0}_\ell\right) \\
	\delta_{\ell,k}(t)&=\delta_k\left(t,\left(\boldsymbol{\Phi}_{\ell,\ell,k}\right)_{k=1}^K,{\bf 0}_\ell\right).
\end{align*}
As it has been shown in \cite{hoydis}, the computation of the first $2J_\ell-1$ derivatives of ${\bf T}_\ell(t)$ and $\delta_{\ell,k}(t)$ at $t=0$, which we denote by  ${\bf T}_{\ell}^{(n)}$ and $\delta_{\ell,k}^{(n)}$, can be performed  using the iterative Algorithm~1, which we provide in Appendix~\ref{app:alogorithm}. These derivatives ${\bf T}_{\ell}^{(n)}$ and $\delta_{\ell,k}^{(n)}$ play a key role in the asymptotic expressions for the SINRs. We are now in a position to state our main results.

\begin{theorem}
	\label{th:main}
	Assume that Assumptions~A-\ref{ass:channel} and A-\ref{ass:correlation} hold true.
Let $\overline{X}_{j,m}(t)$ and $\overline{Z}_{\ell,j,m}(t)$ be
\begin{align*}
\overline{X}_{j,m}(t)&=\frac{\delta_{j,m}(t)}{1+t\delta_{j,m}(t)}\\
	\overline{Z}_{\ell,j,m}(t)&=\frac{1}{K}\tr \big( {\bf R}_{\ell,j,m}{\bf T}_\ell(t) \big)-\frac{t\left|\frac{1}{K}\tr \big(\bPhi_{\ell,j,m}{\bf T}_\ell(t) \big) \right|^2}{1+t\delta_{\ell,m}(t)}.
\end{align*}
Then, in the asymptotic regime defined by  Assumption~A-\ref{ass:regime}, we have
\begin{align*}
&\mathbb{E}\left[X_{j,m}(t)\right]-\overline{X}_{j,m}(t)\xrightarrow[M,K\to+\infty]{}0 \\
&\mathbb{E}\left[Z_{\ell,j,m}(t)\right]-\overline{Z}_{\ell,j,m}(t) \xrightarrow[M,K\to+\infty]{}0.
\end{align*}
Moreover, for every fixed $n$, we have that  ${\rm var}(X_{j,m}^{(n)})=o(1)$.
\end{theorem}
\begin{proof}
	The proof is given in Appendix~\ref{app:main_eq}.
\end{proof}
\begin{corollary}
	\label{corollary:derivation}
Assume the setting of Theorem~\ref{th:main}. Then, in the asymptotic regime we have:
\begin{align*}
&\mathbb{E}\left[X_{j,m}^{(n)}\right] -\overline{X}_{j,m}^{(n)}\xrightarrow[M,K\to+\infty]{} 0 \\
&\mathbb{E}\left[Z_{\ell,j,m}^{(n)}\right] -\overline{Z}_{\ell,j,m}^{(n)}\xrightarrow[M,K\to+\infty]{} 0
\end{align*}
where $\overline{X}_{j,m}^{(n)}$ and $\overline{Z}_{\ell,j,m}^{(n)}$ are the derivatives of $\overline{X}(t)$ and $\overline{Z}_{\ell,j,m}(t)$ with respect to $t$ at $t=0$.
\end{corollary}
\begin{proof}
The proof is given in Appendix~\ref{app:derivation}.
\end{proof}
Theorem~\ref{th:main} provides the tools to calculate the derivatives of ${X}_{j,m}$ and ${Z}_{\ell,j,m}$ at $t=0$, in a recursive manner.
	
Now, denote by  $\overline{X}_{j,m}^{(0)}$  and $\overline{Z}_{\ell,j,m}^{(0)}$ the deterministic quantities given by
	\begin{align*}
		\overline{X}_{j,m}^{(0)}&=\frac{1}{K}\tr (\bPhi_{j,j,m}) \\
		\overline{Z}_{\ell,j,m}^{(0)}&=\frac{1}{K}\tr ({\bf R}_{\ell,j,m}).
	\end{align*}
We can now iteratively compute the deterministic sequences $\overline{X}_{j,m}^{(n)}$ and $\overline{Z}_{\ell,j,m}^{(n)}$ as
	\begin{align*}
		\overline{X}_{j,m}^{(n)}&=-\sum_{k=1}^n {n \choose k} k \overline{X}_{j,m}^{(k-1)}\delta_{j,m}^{(n-k)}+\delta_{j,m}^{(n)} \\
		\overline{Z}_{\ell,j,m}^{(n)}&=\frac{1}{K}\tr \left( {\bf R}_{\ell,j,m}{\bf T}_\ell^{(n)} \right)-\sum_{k=0}^n k {n \choose k} \delta_{l,m}^{(n-k)}\overline{Z}_{\ell,j,m}^{(k-1)} \\
		 &\hspace{-5mm}+\sum_{k=0}^n k {n\choose k} \delta_{l,m}^{(n-k)}\frac{1}{K}\tr \left( {\bf R}_{\ell,j,m}{\bf T}_\ell^{(k-1)} \right) \\
 &\hspace{-5mm}-\sum_{k=0}^n k{n\choose k} \frac{1}{K}\tr \left( \bPhi_{\ell,j,m}{\bf T}_\ell^{(k-1)} \right) \frac{1}{K} \tr \left( \bPhi_{\ell,j,m}{\bf T}_\ell^{(n-k)} \right).
	\end{align*}
Then, from Theorem~\ref{th:main}, we have
\begin{align*}
	&\E [{X}_{j,m}^{(n)}] -\overline{X}_{j,m}^{(n)}\xrightarrow[M,K\to+\infty]{}0,\\
	&\E [{Z}_{\ell,j,m}^{(n)} ]-\overline{Z}_{\ell,j,m}^{(n)}\xrightarrow[M,K\to+\infty]{}0 .
\end{align*}

Plugging the deterministic equivalent of Theorem~\ref{th:main} into \eqref{eq:a_jm} and \eqref{eq:B_ljm}, we get the following corollary.
\begin{corollary}
Let $\overline{\bf a}_{j,m}$ be the vector with elements
	$$
	\left[\overline{\bf a}_{j,m}\right]_{n}=\frac{(-1)^n}{n!}\overline{X}_{j,m}^{(n)}, \ \ n\in \left\{0,\ldots,J_j-1\right\}
	$$
and $\overline{\bf B}_{\ell,j,m}$ the $J_\ell\times J_\ell$ matrix with elements
	$$
	\left[\overline{\bf B}_{\ell,j,m}\right]_{n,p}=\frac{(-1)^{n+p+1}}{(n+p+1)!}\overline{Z}_{\ell,j,m}^{n+p+1},  \ \  n,p \in \left\{0,\ldots,J_\ell-1\right\}.
	$$
Then,
	$$
	\max_{\ell,j,m}\left(\E \left[ \|\overline{\bf B}_{\ell,j,m}-{\bf B}_{\ell,j,m}\| \right],\E \left[\|{\bf a}_{j,m}-\overline{\bf a}_{j,m}\|\right] \right)\xrightarrow[M,K\to+\infty]{}0.
	$$
\end{corollary}
This corollary gives asymptotic equivalents of ${\bf a}_{j,m}$ and ${\bf B}_{\ell,j,m}$, which are the random quantities, that appear in the SINR expression in \eqref{eq:SINR-based-on-TPE}. Hence, we can use these asymptotic equivalents to obtain an asymptotic equivalent of the SINR for all UTs in every cell.

\subsection{Optimization of the System Performance}
\label{sec:optimum}

The previous section developed deterministic equivalents of the SINR at each UT in the multi-cell system, as a function of the polynomial coefficients $\left\{w_{j,\ell},\ell\in\left[1,L\right], j\in\left[0,J_\ell-1\right]\right\}$ of the TPE precoding applied in each of the $L$ cells. These coefficients can be selected arbitrarily, but should not be functions of any instantaneous CSI---otherwise the low complexity properties are not retained. Furthermore, the coefficients need to be scaled such that the transmit power constraints
\begin{equation} \label{eq:constraint-TPE}
\frac{1}{K}\tr \left( {\bf G}_{\ell,\rm TPE}{\bf G}_{\ell,\rm TPE}^{\mbox{\tiny H}} \right) =P_{\ell}
\end{equation}
are satisfied in each cell $\ell$. By plugging the TPE precoding expression from \eqref{eq_TPE-definition} into  \eqref{eq:constraint-TPE}, this implies 
\begin{equation}
	\frac{1}{K}\sum_{n=0}^{J_\ell-1}\sum_{m=0}^{J_\ell-1}w_{n,\ell} w_{m,\ell}^* \left(\frac{ \widehat{\bf H}_{\ell,\ell} \widehat{\bf H}_{\ell,\ell}^{\mbox \tiny H}}{K}\right)^{n+m+1} =P_{\ell}.
\label{eq:constraint}
\end{equation}

In this section, we optimize the coefficients to maximize a general metric of the system performance. To facilitate the optimization, we use the asymptotic equivalents of the SINRs developed in this paper and apply the corresponding asymptotic analysis in order to replace the constraint \eqref{eq:constraint} with its asymptotically equivalent condition
\begin{equation}
	{\bf w}_\ell^{\mbox{\tiny T}}\overline{\bf C}_\ell{\bf w}_\ell=P_{\ell}, \ \ \ell\in\left\{1,\ldots,L\right\}
\label{eq:constraint-approx}
\end{equation}
where $\left[\overline{\bf C}_\ell\right]_{n,m}=\frac{(-1)^{n+m+1}}{(n+m+1)!}\frac{1}{K}\tr ({\bf T}_\ell^{(n+m+1)})$ for all $1 \leq n \leq L$ and $1 \leq m \leq L$.

The performance metric in this section is the weighted max-min fairness, which can provide a good balance between system throughput, user fairness, and computational complexity \cite{Bjornson2013d}.\footnote{Other performance metrics are also possible, but the weighted max-min fairness has often relatively low computational complexity and can be used as a building stone for maximizing other metrics in an iterative fashion \cite{Bjornson2013d}.} This means, that we maximize the minimal value of $\frac{\log_2(1+\gamma_{j,m})}{\nu_{j,m}}$, where the user-specific weights $\nu_{j,m} >0$ are larger for users with high priority (e.g., with favorable channel conditions). Using deterministic equivalents, the corresponding optimization problem is
\begin{equation}
\begin{split}
\maximize{{\bf w}_1,\ldots,{\bf w}_{L}} & \, \min_{\begin{subarray}{l} j \in \left[1,L\right] \\ m\in \left[1,K\right]\end{subarray}}
\frac{1}{\nu_{j,m}} \times \\ & \,\, \log_2 \Bigg( 1 +
\frac{{\bf w}_j^{\mbox{\tiny T}}\overline{\bf a}_{j,m}\overline{\bf a}_{j,m}^{\mbox{\tiny H}}{\bf w}_j}{\displaystyle{\sum_{\ell=1}^{L}}{\bf w}_\ell^{\mbox{\tiny T}}\overline{\bf B}_{\ell,j,m}{\bf w}_\ell-{\bf w}_j^{\mbox{\tiny T}}\overline{\bf a}_{j,m}\overline{\bf a}_{j,m}^{\mbox{\tiny H}}{\bf w}_j}  \Bigg)\\
\mathrm{subject} \,\, \mathrm{to} & \quad \, {\bf w}_\ell^{\mbox{\tiny T}}\overline{\bf C}_\ell{\bf w}_\ell=P_{\ell}, \ \ \ell\in\left\{1,\ldots,L\right\}.
\label{eq:prob}
\end{split}
\end{equation}

This problem has a similar structure as the \emph{joint max-min fair beamforming} problem previously considered in \cite{Karipidis2008a} within the area of multi-cast beamforming communications with several separate user groups. The analogy is the following: The users in cell $j$ in our work corresponds to the $j$th multi-cast group in \cite{Karipidis2008a}, while the coefficients ${\bf w}_j$ in \eqref{eq:prob} correspond to the multi-cast beamforming to group $j$ in \cite{Karipidis2008a}. The main difference is that our problem \eqref{eq:prob} is more complicated due to the structure of the power constraints, the negative sign of the second term in the denominators of the SINRs, and the user weights. Nevertheless, the tight mathematical connection between the two problems implies, that \eqref{eq:prob} is an NP-hard problem because of \cite[Claim 2]{Karipidis2008a}. One should therefore focus on finding a sensible approximate solution to \eqref{eq:prob}, instead of the global optimum.

Approximate solutions to \eqref{eq:prob} can be obtained by well-known techniques from the multi-cast beamforming literature (e.g., \cite{Sidiropoulos2006,Karipidis2008a,Gershman2010a}). For the sake of brevity, we only describe the approximation approach of semi-definite relaxation in this section. To this end we note, we write \eqref{eq:prob} on its equivalent epigraph form
\begin{align} \label{eq:prob2}
\maximize{{\bf w}_1,\ldots,{\bf w}_{L},\xi} & \quad \xi \\ \notag
	\mathrm{subject} \,\, \mathrm{to}  & \quad \, \tr \left( \overline{\bf C}_\ell {\bf w}_\ell {\bf w}_\ell^{\mbox{\tiny T}} \right) =P_{\ell}, \ \ \ell\in\left\{1,\ldots,L\right\} \\ \notag
& \!\!\!\!\!\!\!\!\!\!\!\!\!\!\!\!\!\!\!\!\!\!\!\!  \frac{ \overline{\bf a}_{j,m}^{\mbox{\tiny H}}{\bf w}_j  {\bf w}_j^{\mbox{\tiny T}} \overline{\bf a}_{j,m}  }{\displaystyle{\sum_{\ell=1}^{L}} \tr\left(  \overline{\bf B}_{\ell,j,m}{\bf w}_\ell {\bf w}_\ell^{\mbox{\tiny T}} \right)-   \overline{\bf a}_{j,m}^{\mbox{\tiny H}} {\bf w}_j {\bf w}_j^{\mbox{\tiny T}}  \overline{\bf a}_{j,m}  } \geq 2^{ \nu_{j,m} \xi } \!-\! 1
 \quad \forall j, m
\end{align}
where the auxiliary variable $\xi$ represents the minimal weighted rate among the users. If we substitute the positive semi-definite rank-one matrix ${\bf w}_\ell {\bf w}_\ell^{\mbox{\tiny T}} \in \mathbb{C}^{J_{\ell} \times J_{\ell}}$ for a positive semi-definite matrix ${\bf W}_\ell \in \mathbb{C}^{J_{\ell} \times J_{\ell}}$ of arbitrary rank, we obtain the following tractable relaxed problem
\begin{equation}
\begin{split}
\maximize{{\bf W}_1,\ldots,{\bf W}_{L},\xi} & \quad \xi \\
	\mathrm{subject} \,\, \mathrm{to}  & \quad \, {\bf W}_{\ell} \succeq {\bf 0}, \quad \tr \left( \overline{\bf C}_\ell {\bf W}_\ell \right) =P_{\ell}, \ \ \ell\in\left\{1,\ldots,L\right\} \\
& \!\!\!\!\!\!\!\!\!\!\!\!\!\!\!\!\!\!\!\!\!\!\!\!  \frac{  \overline{\bf a}_{j,m}^{\mbox{\tiny H}}{\bf W}_j  \overline{\bf a}_{j,m}  }{\displaystyle{\sum_{\ell=1}^{L}} \tr\left(  \overline{\bf B}_{\ell,j,m}{\bf W}_\ell  \right)-  \overline{\bf a}_{j,m}^{\mbox{\tiny H}} {\bf W}_j \overline{\bf a}_{j,m} } \geq 2^{ \nu_{j,m} \xi } \!-\! 1 \quad \forall j, m.
\label{eq:prob-relaxation}
\end{split}
\end{equation}
This is a so-called semi-definite relaxation of the original problem \eqref{eq:prob}. Interestingly, for any fixed value on $\xi$, \eqref{eq:prob-relaxation} is a convex semi-definite optimization problem because the power constraints are convex and the SINR constraints can be written in the convex form $\overline{\bf a}_{j,m}^{\mbox{\tiny H}}{\bf W}_j  \overline{\bf a}_{j,m} \geq (2^{ \nu_{j,m} \xi } \!-\! 1) \big( \sum_{\ell=1}^{L} \tr\left(  \overline{\bf B}_{\ell,j,m}{\bf W}_\ell  \right)-  \overline{\bf a}_{j,m}^{\mbox{\tiny H}} {\bf W}_j \overline{\bf a}_{j,m} \big)$. Hence, we can solve \eqref{eq:prob-relaxation} by standard techniques from convex optimization theory for any fixed $\xi$ \cite{Boyd2004a}. In order to also find the optimal value of $\xi$, we note that the SINR constraints become stricter as $\xi$ grows and thus we need to find the largest value for which the SINR constraints are still feasible. This solution process is formalized by the following theorem.

\begin{theorem} \label{th:bisection-solution}
Suppose we have an upper bound $\xi_{\max}$ on the optimum of the problem \eqref{eq:prob-relaxation}. The optimization problem can then be solved by line search over the range $\mathcal{R}=[0,\xi_{\max}]$. For a given value $\xi^{\star} \in \mathcal{R}$, we need to solve the convex feasibility problem
\begin{equation} \label{eq:feasibility-problem}
\begin{split}
\mathrm{find}\,\, & \,\, {\bf W}_1 \succeq {\bf 0} ,\ldots,{\bf W}_{L} \succeq {\bf 0} \\
	\mathrm{subject}\,\,\mathrm{to}\,\, & \,\, \tr \left( \overline{\bf C}_\ell {\bf W}_\ell \right) =P_{\ell}, \ \ \ell\in\left\{1,\ldots,L\right\} \\
& \!\!\!\!\!\!\!\!\!\!\!\!\!\!\!\!\!\!\!\!\!\!\!\!\!\!  \frac{2^{ \nu_{j,m} \xi^{\star} } \!-\! 1  }{2^{ \nu_{j,m} \xi^{\star} } } \displaystyle{\sum_{\ell=1}^{L}} \tr\left(  \overline{\bf B}_{\ell,j,m}{\bf W}_\ell  \right)  -  \overline{\bf a}_{j,m}^{\mbox{\tiny H}} {\bf W}_j \overline{\bf a}_{j,m}  \leq 0   \quad \forall j, m.
\end{split}
\end{equation}
If this problem is feasible, all $\tilde{\xi} \in \mathcal{R}$ with $\tilde{\xi}< \xi^{\star}$ are removed. Otherwise, all $\tilde{\xi} \in \mathcal{R}$ with $\tilde{\xi} \geq \xi^{\star}$ are removed.
\end{theorem}
\begin{IEEEproof}
This theorem follows from identifying \eqref{eq:prob-relaxation} as a quasi-convex problem (i.e., it is a convex problem for any fixed $\xi$ and the feasible set shrinks with increasing $\xi$) and applying any conventional line search algorithms (e.g., the bisection algorithm \cite[Chapter 4.2]{Boyd2004a}).
\end{IEEEproof}

Based on Theorem ~\ref{th:bisection-solution}, we devise the following algorithm based on conventional bisection line search.

\begin{algorithm}[H] \label{algorithm:bisection}
\caption{Bisection algorithm that solves \eqref{eq:prob-relaxation}}
\begin{algorithmic}
    \State Set $\xi_{\min}=0$ and initiate the upper bound $\xi_{\max}$
    \State Select a tolerance $\varepsilon>0$
    \While{$\xi_{\max}-\xi_{\min}> \varepsilon$}
    \State $\xi^{\star} \gets \frac{\xi_{\max}+\xi_{\min}}{2}$ 
    \State Solve \eqref{eq:feasibility-problem} for $\xi^{\star}$
    \If{problem \eqref{eq:feasibility-problem} is feasible}
        \State $\xi_{\min} \gets \xi^{\star}$
        \Else $\,\,\,\xi_{\max} \gets \xi^{\star}$
    \EndIf
    \EndWhile
    \State Output: $\xi_{\min}$ is now less than $\varepsilon$ from the optimum to \eqref{eq:prob-relaxation}
\end{algorithmic}
\end{algorithm}

In order to apply Algorithm \ref{algorithm:bisection}, we need to find a finite upper bound $\xi_{\max}$ on the optimum of \eqref{eq:prob-relaxation}. This is achieved by further relaxation of the problem. For example, we can remove the inter-cell interference and maximize the SINR of each user $m$ in each cell $j$ by solving the problem
\begin{equation}
\begin{split}
\maximize{{\bf w}_j} & \,\,\,  \frac{1}{\nu_{j,m}} \log_2 \left(1 + \frac{{\bf w}_j^{\mbox{\tiny T}}\overline{\bf a}_{j,m}\overline{\bf a}_{j,m}^{\mbox{\tiny H}}{\bf w}_j}{ {\bf w}_j^{\mbox{\tiny T}}\overline{\bf B}_{j,j,m}{\bf w}_j-{\bf w}_j^{\mbox{\tiny T}}\overline{\bf a}_{j,m}\overline{\bf a}_{j,m}^{\mbox{\tiny H}}{\bf w}_j} \right)\\
\mathrm{subject} \,\, \mathrm{to} & \quad \, {\bf w}_j^{\mbox{\tiny T}}\overline{\bf C}_j{\bf w}_j=P_{j}.
\label{eq:prob-single-user}
\end{split}
\end{equation}
This is essentially a generalized eigenvalue problem and therefore solved by scaling the vector ${\bf q}_{j,m} = (\overline{\bf B}_{j,j,m} - \overline{\bf a}_{j,m} \overline{\bf a}_{j,m})^{-1} \overline{\bf a}_{j,m}$ to satisfy the power constraint. We obtain a computationally tractable upper bound $\xi_{\max}$ by taking the smallest of the relaxed SINR among all the users:
\begin{equation}
\xi_{\max} = \min_{j,m} \,\, \frac{\log_2 \left( 1 + \overline{\bf a}_{j,m}^{\mbox{\tiny H}}  (\overline{\bf B}_{j,j,m} - \overline{\bf a}_{j,m} \overline{\bf a}_{j,m})^{-1} \overline{\bf a}_{j,m} \right)}{\nu_{j,m}}.
\end{equation}

The solution to the relaxed problem in \eqref{eq:prob-relaxation} is a set of matrices ${\bf W}_1,\ldots,{\bf W}_L$ that, in general, can have ranks greater than one. In our experience, the rank is indeed one in many practical cases, but when the rank is larger than one we cannot apply the solution directly to the original problem formulation in \eqref{eq:prob}. A standard approach to obtain rank-one approximations is to select the principal eigenvectors of ${\bf W}_1,\ldots,{\bf W}_L$ and scale each one to satisfy the power constraints in \eqref{eq:constraint} with equality.

As mentioned in the proof of Theorem~\ref{th:bisection-solution}, the optimization problem in \eqref{eq:prob-relaxation} belongs to the class of quasi-convex problems. As such, the computational complexity scales polynomially with the number of UTs $K$ and the TPE orders $J_1,\ldots,J_L$. It is important to note that the number of base station antennas $M$ has no impact on the complexity. The exact number of arithmetic operation depends strongly on the choice of the solver algorithm (e.g., interior-point methods \cite{cvx}) and if the implementation is problem-specific or designed for general purposes. As a rule-of-thumb, polynomial complexity means that the scaling is between linear and cubic in the parameters \cite{Laurent2005a}. In any case, the complexity is prohibitively large for real-time computation, but this is not an issue since the coefficients are only functions of the statistics and not the instantaneous channel realizations. In other words, the coefficients for a given multi-cell setup can be computed offline, e.g., by a central node or distributively using decomposition techniques \cite{Palomar2006a}. Even if the channel statistics would change with time, this happens at a relatively slow rate (as compared to the channel realizations), which makes the complexity negligible compared the precoding computations \cite{Kammoun2014a}. Furthermore, we note that the same coefficients can be used for each subcarrier in a multi-carrier system, as the channel statistics are essentially the same across all subcarriers, even though the channel realizations are different due to the frequency-selective fading.

\begin{remark}[User weights that mimic RZF precoding] \label{remark:mimic-RZF}
The user weights $\nu_{j,m}$ can be selected in a variety of ways, resulting in different performance at each UT. Since the main focus of TPE precoding is to approximate RZF precoding, it makes sense to select the user weights to push the performance towards that of RZF precoding. This is achieved by selecting $\nu_{j,m}$ as the rate that user $m$ in cell $j$ would achieve under RZF precoding for some regularization parameters $\varphi_j$ (which, preferably, should be chosen approximately optimal), or rather the deterministic equivalent of this rate in the large-($M,K$) regime; see
Theorem~\ref{th:theorem_rzf} in Section~\ref{sec:review-RZF} for a review of these deterministic equivalents. The optimal $\xi$ from Theorem~\ref{th:bisection-solution} can then be interpreted as the fraction of the RZF precoding performance that is achieved by TPE precoding.
\end{remark}

\section{Simulation Example} \label{sec:simulations}

This section provides a numerical validation of the proposed TPE precoding in a practical deployment scenario.
We consider a three-sector site composed of $L=3$ cells and BSs; see Fig.~\ref{fig:cell}. Similar to the channel model presented in \cite{caire-12}, we assume that the UTs in each cell are divided into $G=2$ groups. UTs of a group share approximatively the same location and statistical properties. We assume that the groups are uniformly distributed in an annulus with an outer radius of $250 \, \mathrm{m}$ and an inner radius of $35 \, \mathrm{m}$, which is compliant with a future LTE urban macro deployment \cite{status-report}.

\begin{figure}
	\begin{center}
		\includegraphics[scale=0.5]{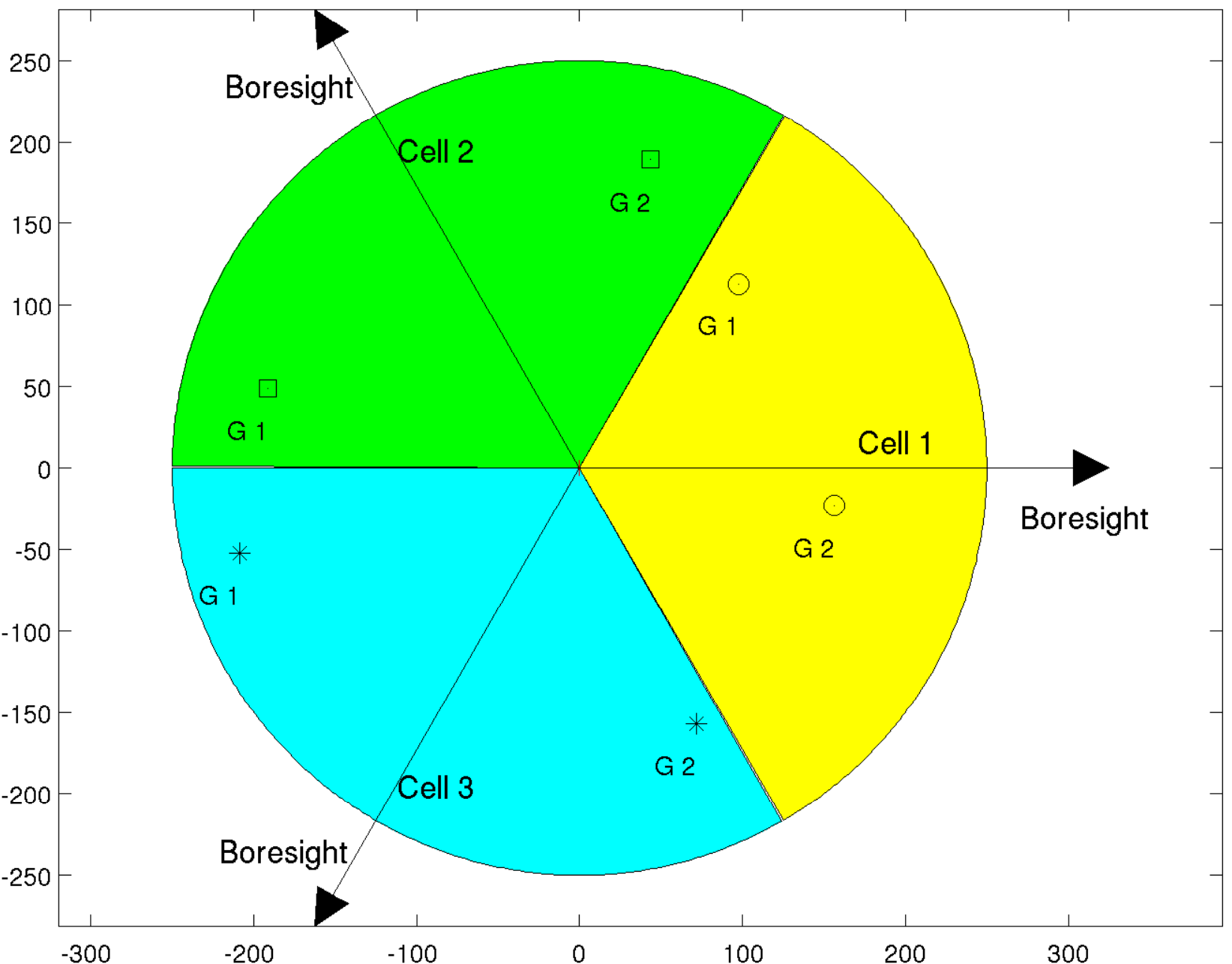}
	\caption{Illustration of the three-sector site deployment with $L=3$ cells considered in the simulations.}
	\label{fig:cell}
	\end{center}
\end{figure}

The pathloss between UT $m$ in group $g$ of cell $j$ and cell $\ell$ follows the same expression as in \cite{caire-12} and is given by
\begin{equation*}
	\mathrm{PL}(d_{\ell,j,m})=\frac{1}{1+(\frac{d_{\ell,j,m}}{d_0})^\delta}
\end{equation*}
where $\delta=3.7$ is the pathloss exponent and $d_0=30 \, \mathrm{m}$ is the reference distance.
Each base station is equipped with an horizontal linear array of $M$ antennas. The radiation pattern of each antenna is
$$
A(\theta)=-\min \left( 12\left(\frac{\theta}{\theta_{3dB}}\right)^2,30 \right) \quad \mathrm{dB}
$$
where $\theta_{3\mathrm{dB}}=70$ degrees and $\theta$ is measured with respect to the BS boresight.
We consider a similar channel covariance model as the one-ring model described in Remark~\ref{remark:correlation}. The only difference is that we scale  the covariance matrix in \eqref{eq:correlation} by the pathloss and the antenna gain:
	 \begin{align*}
		 \left[{\bf R}_{\ell,j,m}\right]_{u,v}=& \frac{10^{A(\theta_{\ell,j,g})/10} \mathrm{PL}(d_{\ell,j,m})}{2\Delta_{\ell,j,g}} \times \\
		 &\qquad \int_{-\Delta_{\ell,j,g}+\theta_{\ell,j,g}}^{\Delta_{\ell,j,g}+\theta_{\ell,j,g}} e^{\jmath d(u-v)\sin\alpha}d\alpha \\
	  &  (\textnormal{user } m \textnormal{ is in group } g  \textnormal{ of cell } j).
\end{align*}
We assume that each BS has acquired imperfect CSI from uplink pilot transmissions with $\rho_{\rm tr} = 15 \, {\rm dB}$. In the downlink, we assume for simplicity that all BSs use the same normalized transmit power of $1$ with $\rho_{\rm dl}=\frac{P}{\sigma^2}=10 \, {\rm dB}$.

The objective of this section is to compare the network throughput of the proposed TPE precoding with that of conventional RZF precoding. To make a fair comparison, the coefficients of the TPE precoding are optimized as described in Remark~\ref{remark:mimic-RZF}. More specifically, each user weight $\nu_{j,m}$ in the semi-definite relaxation problem \eqref{eq:prob} is set to the asymptotic rate that the same user would achieve using RZF precoding. Consequently, the relative differences in network throughput that we will observe in this section hold approximately also for the achievable rate of each UT.

Using Monte-Carlo simulations, we show in Fig.~\ref{fig:network} the average rate per UT, which is defined as
$$
\frac{1}{KL}\sum_{j=1}^L\sum_{m=1}^K \mathbb{E}\left[\log_2\left(1+\gamma_{j,m}\right)\right].
$$
We consider a scenario with $K=40$ users in each cell and different number of antennas at each BS: $M \in \{ 80, \, 160, \, 240,\, 320, \, 400 \}$. The TPE order is the same in all cells: $J = J_j, \forall j$.
As expected, the user rates increase drastically with the number of antennas, due to the higher spatial resolution. The throughput also increases monotonically  with the TPE order $J_j$, as the number of degrees of freedom becomes larger. Note that, if $J_j$ is equal to $4$, increasing $J_j$ leads to a negligible performance improvement that might not justify the increased complexity of having a greater $J_j$. TPE orders of less than $4$ can be relevant in situations when the need for interference-suppression is smaller than usual, for example, if $M/K$ is large (so that the user channels are likely to be near-orthogonal) or when the UTs anticipate small SINRs, due to low performance requirements or large cell sizes. The TPE order is limited only by the available hardware resources and we recall from \cite{Kammoun2014a} that increasing $J_j$ corresponds solely to duplicating already employed circuitry.

Contrary to the single-cell case analyzed in \cite{Kammoun2014a}, where TPE precoding was merely a low-complexity approximation of the optimal RZF precoding, we observe in Fig.~\ref{fig:network} that TPE precoding achieves higher user rates for all $J_j\geq 5$ than the suboptimal RZF precoding (obtained for $\varphi=\sigma^2$). This is due to the optimization of the polynomial coefficients in Section \ref{sec:optimum}, which enables a certain amount of inter-cell coordination, a feature which could not be implemented easily for RZF precoding in multi-cell scenarios.
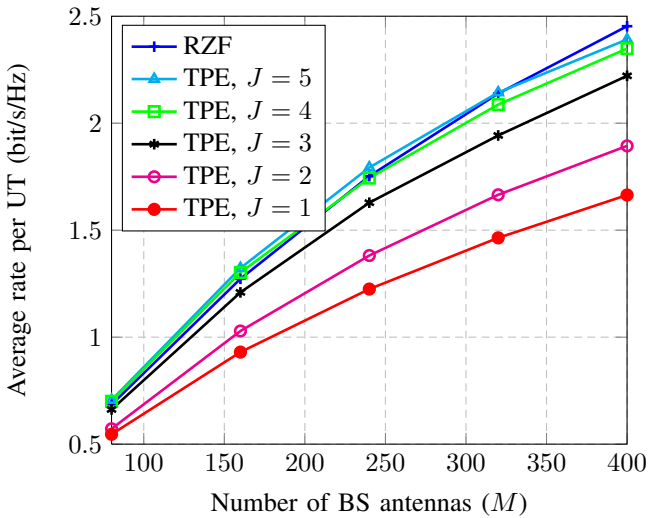
\begin{figure}
	\begin{center}
   \begin{tikzpicture}[scale=1,font=\normalsize]
    \renewcommand{\axisdefaulttryminticks}{4}
    \tikzstyle{every major grid}+=[style=densely dashed]
    \tikzstyle{every axis y label}+=[yshift=-10pt]
    \tikzstyle{every axis x label}+=[yshift=5pt]
    \tikzstyle{every axis legend}+=[cells={anchor=west},fill=white,
        at={(0.02,0.98)}, anchor=north west, font=\normalsize ]
    \begin{axis}[
      xmin=80,
      ymin=0.5,
      xmax=400,
        ymax=2.5,
      grid=major,
      scaled ticks=true,
      xlabel={Number of BS antennas ($M$)},
   			ylabel={Average rate per UT (bit/s/Hz) }			
      ]
      \addplot[color=blue,mark size=2pt,mark =+,line width=1pt] coordinates{
(80.000000,0.684886)(160.000000,1.272637)(240.000000,1.752463)(320.000000,2.138525)(400.000000,2.453235)
            };
      \addlegendentry{ {RZF} }
      
      \addplot[color=cyan,mark size=2pt,mark =triangle,line width=1pt] coordinates{
(80.000000,0.706129)(160.000000,1.322645)(240.000000,1.792772)(320.000000,2.143158)(400.000000,2.391505)
};
      \addlegendentry{ {TPE, $J=5$} }
      
      \addplot[color=green,mark size=2pt,mark =square,line width=1pt] coordinates{
(80.000000,0.700194)(160.000000,1.302508)(240.000000,1.743691)(320.000000,2.087106)(400.000000,2.348614)
};
      \addlegendentry{ {TPE, $J=4$} }
      
      \addplot[color=black,mark size=2pt,mark =asterisk,line width=1pt] coordinates{
(80.000000,0.661990)(160.000000,1.209252)(240.000000,1.628504)(320.000000,1.942589)(400.000000,2.221316)
};
     \addlegendentry{ {TPE, $J=3$} }
     
      \addplot[color=magenta,mark size=2pt,mark =o,line width=1pt] coordinates{
(80.000000,0.571875)(160.000000,1.028783)(240.000000,1.381546)(320.000000,1.665437)(400.000000,1.893923)
};
      \addlegendentry{ {TPE, $J=2$} }
      
      \addplot[color=red,mark size=2pt,mark =*,line width=1pt] coordinates{
(80.000000,0.546061)(160.000000,0.930156)(240.000000,1.224820)(320.000000,1.464116)(400.000000,1.664353)
     };
      \addlegendentry{ {TPE, $J=1$} }
      
    \end{axis}
  \end{tikzpicture}
    \caption{Comparison between conventional RZF precoding and the proposed TPE precoding with different orders $J = J_j, \forall j$.}
  \label{fig:network}
  \end{center}
\end{figure}

From the results of our work in \cite{Kammoun2014a}, we expected that RZF precoding would provide the highest performance if the regularization coefficient is optimized properly. To confirm this intuition, we consider the case where all BSs employ the same regularization coefficient $\varphi$. Fig.~\ref{fig:optimal_rzf} shows the performance of the RZF and TPE precoding schemes as a function of $\varphi$, when $K=100$, $M=250$, and $J=5$. We remind the reader that the TPE precoding scheme indirectly depends on the regularization coefficient $\varphi$, since while solving the optimization problem \eqref{eq:prob-single-user}, we choose the user weights $\nu_{j,m}$  as the asymptotic rates that are achieved by RZF precoding. Fig.~\ref{fig:optimal_rzf} shows that RZF precoding provides the highest performance if the regularization coefficient is chosen very carefully, but TPE precoding is generally competitive in terms of both user performance and implementation complexity.

\begin{figure}
	\begin{center}
   \begin{tikzpicture}[scale=1,font=\normalsize]
    \renewcommand{\axisdefaulttryminticks}{4}
    \tikzstyle{every major grid}+=[style=densely dashed]
    \tikzstyle{every axis y label}+=[yshift=-10pt]
    \tikzstyle{every axis x label}+=[yshift=5pt]
    \tikzstyle{every axis legend}+=[cells={anchor=west},fill=white,
        at={(0.98,0.98)}, anchor=north east, font=\normalsize ]
    \begin{axis}[
      xmin=0,
      ymin=0.8,
      xmax=0.6,
        ymax=1,
      grid=major,
      scaled ticks=true,
      xlabel={ Regularization coefficient $\varphi$},
   			ylabel={Average rate per UT (bit/s/Hz) }			
      ]
\addplot[color=black,mark size=2pt,mark =asterisk,line width=1pt] coordinates{
	     (0.01,0.9547)(0.015,0.9957)(0.1,0.9695)(0.2,0.9005)(0.3,0.8604)(0.4,0.8348)(0.5,0.8164)(0.6,0.8025)
  };
   \addlegendentry{RZF}
 \addplot[color=red,mark size=2pt,mark =square,line width=1pt] coordinates{
	      (0.01,0.9770)(0.015,0.9807)(0.1,0.9806)(0.2,0.9726)(0.3,0.9523)(0.4,0.9340)(0.5,0.9181)(0.6,0.9049)
	      };
	      \addlegendentry{TPE ($J=5$)}
 \addplot[color=magenta,mark size=2pt,mark =x,line width=1pt] coordinates{
(0.010000,0.817519)(0.015000,0.804101)(0.100000,0.903100)(0.200000,0.898846)(0.300000,0.869214)(0.400000,0.854082)(0.500000,0.838469)(0.600000,0.823888)
};
\addlegendentry{TPE ($J=3$)}
\end{axis}
 \end{tikzpicture}
    \caption{Comparison between RZF precoding and TPE precoding for a varying regularization coefficient in RZF.}
  \label{fig:optimal_rzf}
  \end{center}
\end{figure}
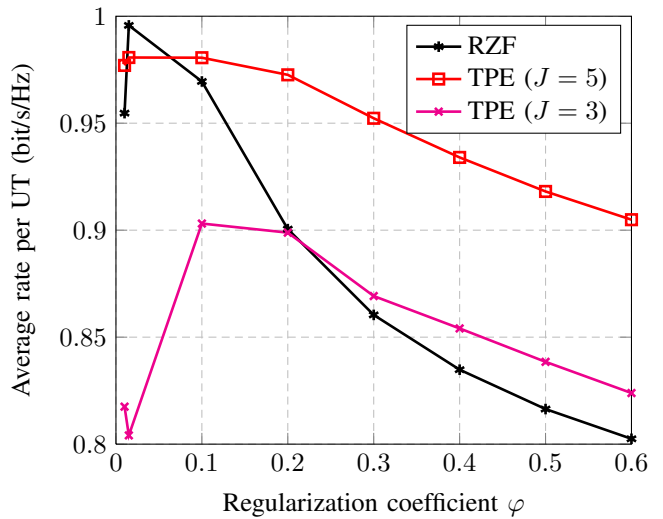

In an additional experiment, we investigate how the performance depends on the effective training SNR ($\rho_{\rm tr}$). Fig.~\ref{fig:training} shows the average rate per UT for $K=100$, $M=250$, $J\in\left\{3,5\right\}$, and $\varphi=0.01$. Note that, as expected, both precoding schemes achieve higher performance as the effective training SNR increases.

\begin{figure}
	\begin{center}
   \begin{tikzpicture}[scale=1,font=\normalsize]
    \renewcommand{\axisdefaulttryminticks}{4}
    \tikzstyle{every major grid}+=[style=densely dashed]
    \tikzstyle{every axis y label}+=[yshift=-10pt]
    \tikzstyle{every axis x label}+=[yshift=5pt]
    \tikzstyle{every axis legend}+=[cells={anchor=west},fill=white,
        at={(0.02,0.98)}, anchor=north west, font=\normalsize ]
    \begin{axis}[
      xmin=0.1,
      ymin=0,
      xmax=12,
        ymax=0.9,
      grid=major,
      scaled ticks=true,
      xlabel={ $\rho_{\rm tr}$},
   			ylabel={Average rate per UT (bit/s/Hz) }			
      ]
      \addplot[color=black,mark size=2pt,mark =triangle,line width=1pt] coordinates{
(0,0.2114)(4,0.3457)(8,0.5581)(12,0.8243)
};
      \addlegendentry{ {TPE, $J=5$} }
      
      \addplot[color=blue,mark size=2pt,mark =+,line width=1pt] coordinates{
      (0,0.1703)(4,0.315)(8,0.529)(12,0.7808)
      };
      \addlegendentry{ {RZF} }
                \end{axis}
  \end{tikzpicture}
    \caption{Comparison between RZF precoding and TPE precoding for a varying effective training SNR $\rho_{\rm tr}$.}
  \label{fig:training}
  \end{center}
\end{figure}
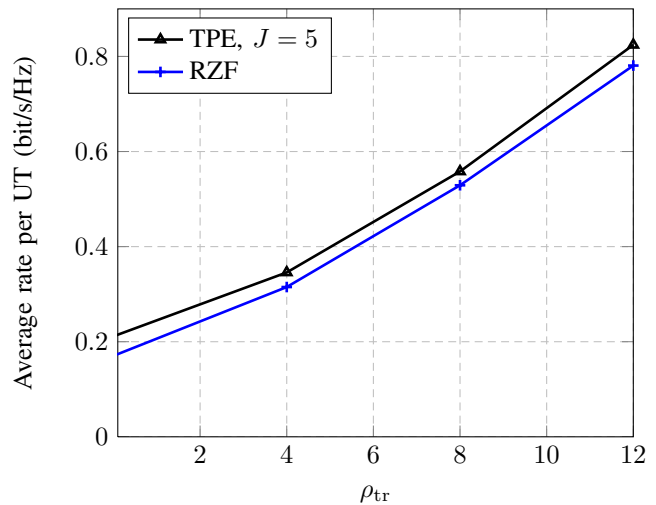

The observed high performance of our TPE precoding scheme is essentially due to the good accuracy of the asymptotic deterministic equivalents. To assess how accurate our asymptotic results are, we show in Fig.~\ref{fig:approximation} the empirical and theoretical UT rates with TPE precoding ($J_j=5$) and RZF precoding with respect to $M$, when $\varphi=\frac{M\sigma^2}{K}$.
 We see that the deterministic equivalents yield a good accuracy even for finite system dimensions. Similar accuracies are also achieved for other regularization factors (recall from Fig.~\ref{fig:network} that the value $\varphi=\frac{M\sigma^2}{K}$ is not optimal), but we chose to visualize a case where the differences between TPE and RZF are large so that the curves are non-overlapping.


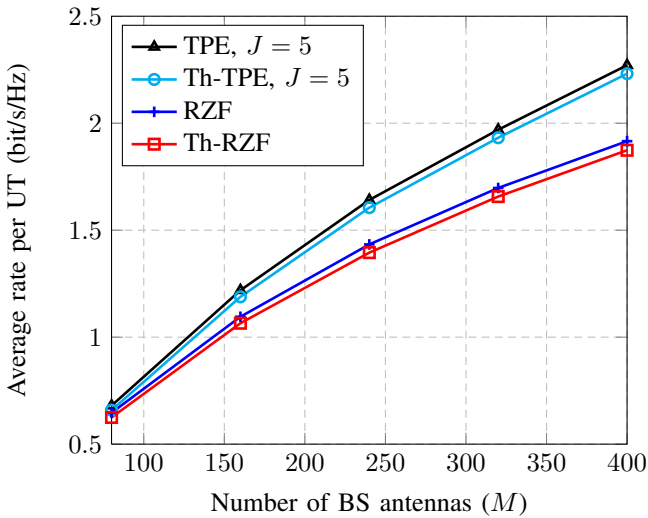
\begin{figure}
	\begin{center}
   \begin{tikzpicture}[scale=1,font=\normalsize]
    \renewcommand{\axisdefaulttryminticks}{4}
    \tikzstyle{every major grid}+=[style=densely dashed]
    \tikzstyle{every axis y label}+=[yshift=-10pt]
    \tikzstyle{every axis x label}+=[yshift=5pt]
    \tikzstyle{every axis legend}+=[cells={anchor=west},fill=white,
        at={(0.02,0.98)}, anchor=north west, font=\normalsize ]
    \begin{axis}[
      xmin=80,
      ymin=0.5,
      xmax=400,
        ymax=2.5,
      grid=major,
      scaled ticks=true,
      xlabel={Number of BS antennas ($M$)},
   			ylabel={Average rate per UT (bit/s/Hz) }			
      ]
      \addplot[color=black,mark size=2pt,mark =triangle,line width=1pt] coordinates{
(80.000000,0.680052)(160.000000,1.219182)(240.000000,1.641179)(320.000000,1.969436)(400.000000,2.270886)
};
      \addlegendentry{ {TPE, $J=5$} }
            \addplot[color=cyan,mark size=2pt,mark =o,line width=1pt] coordinates{
		    (80.000000,0.658750)(160.000000,1.188404)(240.000000,1.605213)(320.000000,1.931940)(400.000000,2.231242)
		};
      \addlegendentry{ {Th-TPE, $J=5$} }
      
      \addplot[color=blue,mark size=2pt,mark =+,line width=1pt] coordinates{
      (80.000000,0.646721)(160.000000,1.095424)(240.000000,1.433480)(320.000000,1.697219)(400.000000,1.915990)
      };
      \addlegendentry{ {RZF} }
      \addplot[color=red,mark size=2pt,mark =square,line width=1pt] coordinates{
(80.000000,0.624880)(160.000000,1.065109)(240.000000,1.395087)(320.000000,1.657076)(400.000000,1.873531)      };
      \addlegendentry{ {Th-RZF } }

    \end{axis}
  \end{tikzpicture}
    \caption{Comparison between the empirical and theoretical user rates. This figure illustrates the asymptotic accuracy of the deterministic approximations.}
  \label{fig:approximation}
  \end{center}
\end{figure}

\section{Conclusion}

This paper generalizes the recently proposed TPE precoder to multi-cell large scale MIMO systems.
This class of precoders originates from the high-complexity RZF precoding scheme by approximating the regularized channel inversion by a truncated polynomial expansion. 

The model includes important multi-cell characteristics, such as user-specific channel statistics, pilot contamination, different TPE orders in different cells, and cell-specific power constraints. We derived  asymptotic SINR expressions, which depend only on channel statistics, that are exploited to optimize the polynomial coefficients in an offline manner. 

The effectiveness of the proposed TPE precoding is illustrated numerically.
Contrary to the single-cell case, where  RZF leads to a near-optimal performance when the regularization coefficient is properly chosen, the use of the RZF precoding in the multi-cell scenario is more delicate. Until now, there is no general rule for the selection of its regularization coefficients.
This enabled us to achieve higher throughput with our TPE precoding for certain scenarios. This is a remarkable result, because TPE precoding therefore has \emph{both} lower complexity and better throughput. This is explained by the use of optimal polynomial coefficients in TPE precoding, while the corresponding optimization of the regularization matrix in RZF precoding has not been obtained so far.

\appendices

\section{Some Useful Results}
 \begin{lemma}[Common inverses of resolvents] \label{lemma:woodbury}
     Given any matrix $\widehat{\bf H} \in \mathbb{C}^{M\times K}$, let $\widehat{\bf h}_k$ denote its $k$th column and $\widehat{\bf H}_k$ be the matrix obtained after removing the $k$th column from $\widehat{\bf H}$.
      The resolvent matrices of $\widehat{\bf H}$ and $\widehat{\bf H}_k$ are denoted by
\begin{align*}
      {\bf Q}(t)&=\left(\frac{t}{K}\widehat{\bf H}\widehat{\bf H}^{\mbox{\tiny H}}+{\bf I}_M\right)^{-1} \\
      {\bf Q}_k(t)&=\left(\frac{t}{K}\widehat{\bf H}_k\widehat{\bf H}_k^{\mbox{\tiny H}}+{\bf I}_M\right)^{-1}
\end{align*}
respectively. It then holds, that
      \begin{equation*}
      {\bf Q}(t)={\bf Q}_k(t)-\frac{1}{K}\frac{t{\bf Q}_k(t)\widehat{\bf h}_k\widehat{\bf h}_k^{\mbox{\tiny H}}{\bf Q}_k(t)}{1+\frac{t}{K}\widehat{\bf h}_k^{\mbox{\tiny H}}{\bf Q}_k(t)\widehat{\bf h}_k}
      \label{eq:Q_k}
      \end{equation*}
      and also
     \begin{equation}
	     {\bf Q}(t)\widehat{\bf h}_k=\frac{{\bf Q}_k(t)\widehat{\bf h}_k}{1+\frac{t}{K}\widehat{\bf h}_k^{\mbox{\tiny H}}{\bf Q}_k(t)\widehat{\bf h}_k} .
	     \label{eq:Q_k_h}
     \end{equation}
     \end{lemma}

\begin{lemma}[Convergence of quadratic forms \cite{SIL98}]
	\label{lemma:quadratic}
Let ${\bf x}_M=\left[X_1,\ldots,X_M\right]^{\mbox{\tiny T}}$ be an $M\times 1$ vector where the $X_n$ are i.i.d.\ Gaussian complex random variables  with unit variance. Let ${\bf A}_M$ be an $M\times M$ matrix independent of ${\bf x}_M$ whose spectral norm is bounded; that is, there exists $C_A$ such that $\|{\bf A}\|\leq C_A$. Then, for any $p\geq 2$, there exists a constant $C_p$, depending only in $p$, such that
\begin{align}
	& \mathbb{E}_{{\bf x}_M} \left[ \left|\frac{1}{M}{\bf x}_M^{\mbox{\tiny H}}{\bf A}_M{\bf x}_M - \frac{1}{M}\tr ({\bf A}_M) \right|^p \right] \leq \nonumber \\
	& \ \frac{C_p}{M^p}\left(\left(\mathbb{E}|X_1|^4\tr\left({{\bf A}{\bf A}^{\mbox{\tiny H}}}\right)\right)^{p/2} \right.								\left.+\mathbb{E}|X_1|^{2p}\tr\left({\bf A}{\bf A}^{\mbox{\tiny H}}\right)^{p/2}\right)\label{eq:quadratic_form}
\end{align}
where the expectation is taken over the distribution of ${\bf x}_M$.
Noticing, that $\tr\left({{\bf A}{\bf A}^{\mbox{\tiny H}}}\right)\leq M\|{\bf A}\|^2$ and that $\tr\left({{\bf A}{\bf A}^{\mbox{\tiny H}}}\right)^{p/2}\leq M\|{\bf A}\|^{p}$, we obtain the simpler inequality:
$$
\mathbb{E}_{{\bf x}_M} \left[ \left|\frac{1}{M}{\bf x}_M^{\mbox{\tiny H}}{\bf A}_M{\bf x}_M - \frac{1}{M}\tr ({\bf A}_M) \right|^p \right] \leq \frac{C_p^{'}\|{\bf A}\|^p}{M^{p/2}}$$
where $C_p^{'}=C_p\left(\left(\mathbb{E} [ |X_1|^4 ] \right)^{p/2}+\mathbb{E} [ |X_1|^{2p} ] \right)$.
By choosing $p\geq 4$, we thus have that
$$
\frac{1}{M}{\bf x}^{\mbox{\tiny H}}{\bf A}_M{\bf x}-\frac{1}{M} \tr ({\bf A}_M) \xrightarrow[M\to+\infty]{a.s.} 0.
$$
\end{lemma}
\begin{corollary}
\label{corollary:zero_quadratic}
	Let ${\bf A}_M$ be as in Lemma \ref{lemma:quadratic}, and ${\bf x}_M,{\bf y}_M$ be random, mutually independent with complex Gaussian entries of zero mean and variance $1$. Then, for any $p\geq 2$ we have
$$
\mathbb{E} \left[ \left|\frac{1}{M} {\bf y}_M^{\mbox{\tiny H}}{\bf A}_M{\bf x}_M\right|^p \right] =O(M^{-p/2}).
$$
In particular,
$$
\frac{1}{M} {\bf y}_M^{\mbox{\tiny H}}{\bf A}_M{\bf x}_M\xrightarrow[M,K\to+\infty]{a.s.}0.
$$
\end{corollary}
\begin{lemma}[Rank-one perturbation lemma]
\label{lemma:perturbation}
Let ${\bf Q}(t)$ and ${\bf Q}_k(t)$ be the resolvent matrices as defined in Lemma~\ref{lemma:woodbury}. Then, for any matrix ${\bf A}$ we have:
$$
\tr \big( {\bf A}\left({\bf Q}(t)-{\bf Q}_k(t)\right) \big) \leq \|{\bf A}\| .
$$
\end{lemma}
\begin{lemma}[Leibniz formula for the derivatives of a product of functions]
	\label{lemma:derivative}
	Let $t\mapsto f(t)$ and $t\mapsto g(t)$ be two $n$ times differentiable functions. Then, the $n$th derivative of the product $f \cdot g$ is given by
	$$
	\frac{d^n f \cdot g}{dt^n}=\sum_{k=0}^n {n\choose k} \frac{d^k f}{dt^k}\frac{d^{n-k}g}{dt^{n-k}}.
	$$
\end{lemma}
Applying Lemma \ref{lemma:derivative} to the function $t\mapsto tf(t)$, we obtain the following result.
\begin{corollary}
	The $n$th derivative of $t\mapsto tf(t)$ at $t=0$ yields
$$
\lb\frac{d^n tf(t)}{dt^n}\rabs_{t=0}=n\lb\frac{d^{n-1}f}{dt^{n-1}}\rabs_{t=0}.
$$
\end{corollary}
\section{Proof of Theorem \ref{th:main}}
\label{app:main_eq}
The objective of this section is to find deterministic equivalents for $\mathbb{E}\left[X_{j,m}(t)\right]$ and $\mathbb{E}\left[Z_{j,m}(t)\right]$.
 These  quantities involve the resolvent matrix
$$
\bsigma(t,j)=\left(t\frac{\widehat{\bf H}_{j,j}\widehat{\bf H}_{j,j}^{\mbox{\tiny H}}}{K}+{\bf I}_M\right)^{-1}.
$$
For technical reasons, the resolvent matrix $\bsigma_{m}(t,j)$, that is obtained by removing the contribution of vector $\widehat{\bf h}_{j,j,m}$ will  be extensively used.
In particular, if $\widehat{\bf H}_{j,j,-m}$ denotes the matrix $\widehat{\bf H}_{j,j}$ after removing the $m$th column,  $\bsigma_{m}(t,j)$ is given by
$$
\bsigma_{m}(t,j)=\left(t\frac{\widehat{\bf H}_{j,j,-m}\widehat{\bf H}_{j,j,-m}^{\mbox{\tiny H}}}{K}+{\bf I}_M\right)^{-1}.
$$
With this notation on hand, we are now in position to prove Theorem \ref{th:main}.
In the sequel, we will mean by "controlling a certain quantity" the study of its asymptotic behaviour in the asymptotic regime.

\subsection{Controlling $X_{j,m}(t)$ and $Z_{\ell,j,m}(t)$}
Next, we study sequentially the random quantities $X_{j,m}(t)$ and $Z_{\ell,j,m}(t)$.
Using Lemma \ref{lemma:woodbury}, the matrix $\bsigma(t,j)$ writes as
\begin{equation}
\bsigma(t,j)=\bsigma_{m}(t,j)-\frac{t}{K}\frac{\bsigma_m(t,j)\widehat{\bf h}_{j,j,m}\widehat{\bf h}_{j,j,m}^{\mbox{\tiny H}}\bsigma_{m}(t,j)}{1+\frac{t}{K}\widehat{\bf h}_{j,j,m}^{\mbox{\tiny H}}\bsigma_m(t,j)\widehat{\bf h}_{j,j,m}}.
\label{eq:essential}
\end{equation}
Plugging \eqref{eq:essential} into the expression of $X_{j,m}(t)$, we get
\begin{align}
X_{j,m}(t)&=\frac{1}{K}{\bf h}_{j,j,m}^{\mbox{\tiny H}}\bsigma_{m}(t,j)\widehat{\bf h}_{j,j,m}\nonumber\\
			       &-\frac{\frac{t}{K^2}{\bf h}_{j,j,m}^{\mbox{\tiny H}}\bsigma_{m}(t,j)\widehat{\bf h}_{j,j,m}\widehat{\bf h}_{j,j,m}^{\mbox{\tiny H}}\bsigma_{m}(t,j)\widehat{\bf h}_{j,j,m}^{\mbox{\tiny H}}}{1+\frac{t}{K}\widehat{\bf h}_{j,j,m}^{\mbox{\tiny H}}\bsigma_{m}(t,j)\widehat{\bf h}_{j,j,m}}\nonumber\\
		 &=\frac{\frac{1}{K}{\bf h}_{j,j,m}^{\mbox{\tiny H}}\bsigma_m(t,j)\widehat{\bf h}_{j,j,m}}{1+\frac{t}{K}\widehat{\bf h}_{j,j,m}^{\mbox{\tiny H}}\bsigma_m(t,j)\widehat{\bf h}_{j,j,m}}. \label{eq:X_jm}
\end{align}
Since ${\bf h}_{j,j,m}-\widehat{\bf h}_{j,j,m}$ is uncorrelated with $\widehat{\bf h}_{j,j,m}$, we have
$$
\mathbb{E}\left[X_{j,m}(t)\right]=\mathbb{E}\left[\frac{\frac{1}{K}\widehat{\bf h}_{j,j,m}^{\mbox{\tiny H}}\bsigma_m(t,j)\widehat{\bf h}_{j,j,m}}{1+\frac{t}{K}\widehat{\bf h}_{j,j,m}^{\mbox{\tiny H}}\bsigma_m(t,j)\widehat{\bf h}_{j,j,m}}\right].
$$
Using Lemma \ref{lemma:quadratic}, we then prove that
\begin{equation}
\frac{1}{K}\widehat{\bf h}_{j,j,m}^{\mbox{\tiny H}}\bsigma_m(t,j)\widehat{\bf h}_{j,j,m}-\frac{1}{K}\tr \big( \bPhi_{j,j,m}\bsigma_m(t,j) \big) \xrightarrow[M,K\to+\infty]{\mathrm{a.s.}}0.
\label{eq:first_eq}
\end{equation}
Applying the rank one perturbation Lemma \ref{lemma:perturbation},
\begin{equation}
\frac{1}{K}\tr \big( \bPhi_{j,j,m}\bsigma_m(t,j) \big) -\frac{1}{K}\tr \big( \bPhi_{j,j,m} \bsigma(t,j) \big) \xrightarrow[M,K\to+\infty]{\mathrm{a.s.}}0.
\label{eq:second_eq}
\end{equation}
On the other hand, Theorem \ref{th:deterministic_1} implies that
\begin{equation}
\frac{1}{K}\tr \big( \bPhi\bsigma(t,j) \big) -\frac{1}{K} \tr \big( \bPhi_{j,j,m}{\bf T}_j(t) \big) \xrightarrow[M,K\to+\infty]{\mathrm{a.s.}}0.
\label{eq:third_eq}
\end{equation}
Combining  \eqref{eq:first_eq}, \eqref{eq:second_eq}, and \eqref{eq:third_eq}, we obtain the following result:
$$
\frac{1}{K}\widehat{\bf h}_{j,j,m}^{\mbox{\tiny H}}\bsigma_m(t,j)\widehat{\bf h}_{j,j,m} -\delta_{j,m}(t)\xrightarrow[M,K\to+\infty]{\mathrm{a.s.}}0
$$
where we used the fact that $\delta_{j,m}(t)=\frac{1}{K} \tr ( \bPhi_{j,j,m}{\bf T}_j(t) )$.
Since $f:x\mapsto\frac{x}{tx+1}$ is bounded by $\frac{1}{t}$, the dominated convergence theorem \cite{BIL08} allow us to conclude that
$$
\mathbb{E}\left[X_{j,m}(t)\right]-\frac{\delta_{j,m}(t)}{1+t\delta_{j,m}(t)}\xrightarrow[M,K\to+\infty]{}0.
$$
We now move to the control of $\mathbb{E}\left[Z_{j,\ell,m}(t)\right]$. Similarly, we first decompose $\mathbb{E}\left[Z_{\ell,j,m}(t)\right]$, by using Lemma \ref{lemma:woodbury}, as
\begin{align*}
Z_{\ell,j,m}(t)&=\frac{1}{K}{\bf h}_{\ell,j,m}^{\mbox{\tiny H}}\bsigma_m(t,\ell){\bf h}_{\ell,j,m}\\
&- \frac{\frac{t}{K^2}{\bf h}_{\ell,j,m}^{\mbox{\tiny H}}\bsigma_m(t,\ell)\widehat{\bf h}_{\ell,\ell,m}\widehat{\bf h}_{\ell,\ell,m}^{\mbox{\tiny H}}\bsigma_m(t,\ell){\bf h}_{\ell,j,m}}{1+\frac{t}{K}\widehat{\bf h}_{\ell,\ell,m}^{\mbox{\tiny H}}\bsigma_m(t,\ell)\widehat{\bf h}_{\ell,\ell,m}}\\
&\triangleq U_{\ell,j,m}(t)-V_{\ell,j,m}(t).
\end{align*}
Let us begin by treating  $\mathbb{E}\left[U_{\ell,j,m}(t)\right]$. Since ${\bf h}_{\ell,j,m}$ and $\bsigma_m(t,\ell)$ are independent, we have
$$
\mathbb{E}\left[U_{\ell,j,m}(t)\right]=\mathbb{E}\left[\frac{1}{K}\tr \big( {\bf R}_{\ell,j,m}\bsigma_m(t,\ell) \big) \right].
$$
Working out the obtained expression using  \eqref{eq:second_eq} and \eqref{eq:third_eq}, we obtain
$$
\mathbb{E}\left[U_{\ell,j,m}(t)\right]-\frac{1}{K}\tr \big( {\bf R}_{\ell,j,m}{\bf T}_\ell(t) \big) \xrightarrow[M,K\to+\infty]{} 0.
$$
As for the control of $V_{\ell,j,m}$ we need to introduce the following quantities:
$$\beta_{\ell,j,m}=\frac{\sqrt{t}}{K}{\bf h}_{\ell,j,m}^{\mbox{\tiny H}}\bsigma_m(t,\ell)\widehat{\bf h}_{\ell,\ell,m}$$
and
$$\stackrel{o}{\beta}_{\ell,j,m}=\beta_{\ell,j,m}-\mathbb{E}_{h}\left[\beta_{\ell,j,m}\right]$$
where $\mathbb{E}_{h}[\cdot]$ denotes the expectation with respect to vector ${\bf h}_{\ell,k,m}$, $k=1,\ldots,L$.
Let $\alpha_{\ell,m}=\widehat{\bf h}_{\ell,\ell,m}\bsigma_m(t,\ell)\widehat{\bf h}_{\ell,m}$. Then, we have
\begin{align}
\mathbb{E}\left[{V}_{\ell,j,m}(t)\right]&=\mathbb{E}\left[\frac{\left|{\beta}_{\ell,j,m}\right|^2}{1+t\alpha_{\ell,m}}\right]\nonumber \\
&=\mathbb{E}\left[\frac{\left|\mathbb{E}_h\beta_{\ell,j,m}\right|^2}{(1+t\alpha_{\ell,m})}\right] +\mathbb{E}\left[\frac{\left|\mathbb{E}_h\left[\stackrel{o}{\beta}_{\ell,j,m}\right]\right|^2}{1+t\alpha_{\ell,m}} \right] \nonumber \\
&+\mathbb{E}\left[\frac{2\Re\left(\stackrel{o}{\beta}_{\ell,j,m}\mathbb{E}_h\left[\beta_{\ell,j,m}\right]\right)}{1+t\alpha_{l,m}}\right]\label{eq:v_ljm}
\end{align}
where $\Re(\cdot)$ denotes the real-valued part of a scalar. Using Lemma \ref{lemma:quadratic}, we can show that the last terms in the right hand side of \eqref{eq:v_ljm} tend to zero. Therefore,
\begin{align}
\mathbb{E}\left[V_{\ell,j,m}(t)\right]&=\mathbb{E}\left[\frac{t\left|\frac{1}{K}\tr \big( \bPhi_{\ell,j,m}\bsigma_{m}(t,\ell) \big) \right|^2}{1+t\alpha_{\ell,m}}\right]+o(1)\nonumber\\
&\stackrel{(a)}=\mathbb{E}\left[\frac{t\left|\frac{1}{K}{\tr \big( \bPhi_{\ell,j,m}{\bf T}_\ell(t) \big) }\right|^2}{1+t\alpha_{\ell,m}}\right]+o(1) \label{eq:v_ljm_inter}
\end{align}
where $(a)$ follows from that
$$
\mathbb{E}\left[\frac{1}{K}\tr \big( \bPhi_{\ell,j,m}\bsigma_{m}(t,\ell) \big) \right]-\frac{1}{K}\tr \big(\bPhi_{\ell,j,m}{\bf T}_\ell(t) \big) \xrightarrow[M,K\to+\infty]{} 0.
$$
On the other hand, one can prove using \eqref{lemma:quadratic} that
$$
\alpha_{\ell,m} - \delta_{\ell,m}\xrightarrow[M,K\to+\infty]{\mathrm{a.s.}} 0
$$
and as such
\begin{equation}
\mathbb{E}\left[\frac{1}{1+t\alpha_{\ell,m}}\right]-\frac{1}{1+t\delta_{\ell,m}(t)}\xrightarrow[M,K\to+\infty]{} 0.
\label{eq:v_ljm_inter_bis}
\end{equation}
Combining \eqref{eq:v_ljm_inter} and \eqref{eq:v_ljm_inter_bis}, we obtain
$$
\mathbb{E}\left[V_{\ell,j,m}(t)\right]=\frac{t\left|\frac{1}{K}\tr \big( \bPhi_{\ell,j,m}{\bf T}_\ell(t) \big)\right|^2}{1+t\delta_{\ell,m}(t)}+o(1).
$$
Finally, substituting $\mathbb{E}\left[U_{\ell,j,m}(t)\right]$ and $\mathbb{E}\left[V_{\ell,j,m}(t)\right]$ by their deterministic equivalents gives the desired result.
\section{Proof of Corollary \ref{corollary:derivation}}
\label{app:derivation}
From Theorem~\ref{th:main} we have that, $X_{j,m}(t)$ and $Z_{\ell,j,m}(t)$ converge to deterministic equivalents which we denote by $\overline{X}_{j,m}(t)$ and $\overline{Z}_{\ell,j,m}(t)$. Corollary~\ref{corollary:derivation} extends this result to the convergence of the derivatives. Its proof is based on the same techniques used in our work in \cite{Kammoun2014a}. We provide hereafter the adapted proof for sake of completeness. We restrict ourselves to the control of $X_{j,m}^{(n)}$, as $Z_{\ell,j,m}^{(n)}$ can be treated analogously.
First note that  $X_{j,m}(t)-\overline{X}_{j,m}(t)$ is analytic, when extended to $\mathbb{C}\backslash\mathbb{R}_{-}$, where $\mathbb{R}_{-}$ is the set of negative real-valued scalars. As it is almost surely bounded on every compact subset of  $\mathbb{C}\backslash\mathbb{R}_{-}$, Montel's theorem \cite{RUD86} ensures that there exists a converging subsequence that converges to an analytic function. Since the limiting function is zero on $\mathbb{R}_{+}$, it must be zero everywhere because of analyticity. Therefore, from every subsequence one can extract a convergent subsequence, that converges to zero. Necessarily, $X_{j,m}(t)-\overline{X}_{j,m}(t)$  converges to zero for every $t\in \mathbb{C}\backslash\mathbb{R}_{-}$. Due to analyticity of the functions \cite{RUD86}, we also have
\begin{equation}X_{j,m}^{(n)}(t)-\overline{X}_{j,m}^{(n)}(t) \xrightarrow[M,K\to+\infty]{a.s.} 0, \quad \textnormal{for every } t\in\mathbb{C}\backslash\mathbb{R}_{-}.\label{eq:final_result}\end{equation}
To extend the convergence result to $t=0$ we will, in a similar fashion as in \cite{Kammoun2014a}, decompose $X_{j,m}^{(n)}-\overline{X}_{j,m}^{(n)}$ as
$$
X_{j,m}^{(n)}-\overline{X}_{j,m}^{(n)}=\alpha_1+\alpha_2+\alpha_3
$$
where $\alpha_1, \alpha_2$ and $\alpha_3$  are
\begin{align*}
	\alpha_1&=X_{j,m}^{(n)}-X_{j,m}^{(n)}(\eta) \\
	\alpha_2&=X_{j,m}^{(n)}(\eta)-\overline{X}_{j,m}^{(n)}(\eta) \\
	\alpha_3&=\overline{X}_{j,m}^{(n)}(\eta)-\overline{X}_{j,m}^{(n)}.
\end{align*}
Note that $X_{j,m}^{(n)}(\eta)$ and $\overline{X}_{j,m}^{(n)}(\eta)$ are, respectively, the $n$th derivatives of $X_{j,m}(t)$ and $\overline{X}_{j,m}(t)$ at $t=\eta$. We rewrite $\alpha_1$ as
\begin{align*}
	\alpha_1&=\frac{1}{K}{\bf h}_{j,j,m}^{\mbox{\tiny H}}\left({\bf I}-\bsigma(\eta,j)\right)\widehat{\bf h}_{j,j,m}\\
		       &=\frac{\eta}{K}{\bf h}_{j,j,m}^{\mbox{\tiny H}}\frac{\widehat{\bf H}_{j,j}\widehat{\bf H}_{j,j}^{\mbox{\tiny H}}}{K}\bsigma(\eta,j)\widehat{\bf h}_{j,j,m}.
\end{align*}
Therefore,
$$
\left|\alpha_1\right|\leq |\eta|  \left\|\frac{{\bf h}_{j,j,m}}{\sqrt{K}}\right\| \left\|\frac{\widehat{\bf h}_{j,j,m}}{\sqrt{K}}\right\| \left\|\frac{\widehat{\bf H}_{j,j}\widehat{\bf H}_{j,j}^{\mbox{\tiny H}}}{K}\right\| .
$$
Since $\|\frac{{\bf h}_{j,j,m}}{\sqrt{K}}\|$, $\|\frac{\widehat{\bf h}_{j,j,m}}{\sqrt{K}}\|$ and $\|\frac{\widehat{\bf H}_{j,j}\widehat{\bf H}_{j,j}^{\mbox{\tiny H}}}{K}\|$ are almost surely bounded, there exists $M_0$ and a constant $C_0$, such that for all $M\geq M_0$, $\left|\alpha_1\right|\leq C_0\eta$. Hence, for $\eta\leq \frac{\epsilon}{3C_0}$, we have  $\left|\alpha_1\right|\leq\frac{\epsilon}{3}$.
On the other hand, $\overline{X}_{j,m}^{(n)}(t)$ is continuous at $t=0$. So there exists $\eta$ small enough such that $\left|\alpha_3\right|=\left|\overline{X}_{j,m}^{(n)}(\eta)-\overline{X}_{j,m}^{(n)}\right|\leq \frac{\epsilon}{3}$. Finally, Eq.~\eqref{eq:final_result} asserts that there exists $M_1$ such that for any $M\geq M_1$, $\left|\alpha_2\right|\leq \frac{\epsilon}{3}$. Take $M\geq \max(M_0,M_1)$ and $\eta\leq \frac{\epsilon}{3C_0}$, we then have
$$
\left|X_{j,m}^{(n)}-\overline{X}_{j,m}^{(n)}\right|\leq {\epsilon},
$$
thereby establishing
$$
X_{j,m}^{(n)}-\overline{X}_{j,m}^{(n)}\xrightarrow[M,K\to+\infty]{a.s.}0.
$$

\section{Algorithm for computing ${\bf T}_\ell$ and $\delta_{\ell,m}$.}
\label{app:alogorithm}
\begin{algorithm}[H]
\caption{Iterative algorithm for computing the first $D$ derivatives of deterministic equivalents at $t=0$ }
\begin{algorithmic}
	\For {$\ell=1 \to L$}
	\For {$k=1 \to K$}
	\State $\delta_{\ell,k}^{(0)}\gets \frac{1}{K}\tr ({\bPhi}_{\ell,\ell,k})$ 
	\State $g_{\ell,k}^{(0)} \gets 0$
	\State $f_{\ell,k}^{(0)} \gets -\frac{1}{1+g_{\ell,k}^{(0)}}$
	\EndFor
	\State ${\bf T}_\ell^{(0)} \gets {\bf I}_M$
	\State ${\bf Q}^{(0)} \gets {\bf 0}_M$
	\For {$i=1 \to D$}
	\State ${\bf Q}^{(i)}\gets \frac{i}{K}\sum_{k=1}^{K}f_{k}^{(i-1)}{\bPhi}_{\ell,\ell,k}$
	\State ${\bf T}_\ell^{(i)} \gets \displaystyle{\sum_{n=0}^{i-1}\sum_{j=0}^n {i-1\choose n }{n\choose j}{\bf T}_\ell^{(i-1-n)}{\bf Q}^{(n-j+1)}{\bf T}_\ell^{(j)}}$
	\For {$k=1 \to K$}
	\State $f_{\ell,k}^{(i)}\gets \displaystyle{\sum_{n=0}^{i-1}\sum_{j=0}^i} {i-1\choose n}{n \choose j}(i-n) \times \newline \hspace*{30mm} f_{\ell,k}^{(j)} f_{\ell,k}^{(i-j)}\delta_{\ell,k}^{(i-1-n)}$
	\State $g_{\ell,k}^{(i)} \gets i\delta_{\ell,k}^{(i-1)}$
	\State $\delta_{\ell,k}^{(i)}\gets \frac{1}{K}\tr ({\bPhi}_{\ell,\ell,k}{\bf T}_\ell^{(i)})$
	\EndFor
	\EndFor
	\EndFor
\end{algorithmic}
\end{algorithm}
    \bibliographystyle{IEEEbib}
\bibliography{IEEEabrv,IEEEconf,./tutorial_RMT}

\begin{IEEEbiographynophoto}
{Abla Kammoun} was born in Sfax, Tunisia. She received the engineering degree in signal and systems from the Tunisia Polytechnic School, La Marsa, and the Master's degree and the Ph.D. degree in digital communications from Telecom Paris Tech [then Ecole Nationale Supérieure des Télécommunications (ENST)]. From June 2010 to April 2012, she has been a Postdoctoral Researcher in the TSI Department, Telecom Paris Tech. Then she has been at Supélec at the Alcatel-Lucent Chair on Flexible Radio until December 2013. Currently, she is a Postodoctoral fellow at KAUST university. Her research interests include performance analysis, random matrix theory, and semi-blind channel estimation.
\end{IEEEbiographynophoto}
\begin{IEEEbiographynophoto}{
Axel M\"uller} (S'11)
received a Dipl.-Ing.(FH) degree in electrical engineering from the university of applied sciences Ulm and a M.Sc. degree in communications engineering from the University of Ulm, Germany, in 2009 and 2011, respectively.
Since 2011, he is pursuing a Ph.D. degree at Supélec, Gif-sur-Yvette, France. His research interests are in the area of large random matrix theory and information theory, especially in the context of applications to interference mitigation techniques in dense heterogeneous cellular networks.
\end{IEEEbiographynophoto}
\begin{IEEEbiographynophoto}{
Emil Bj\"ornson} was born in Malmö, Sweden, in 1983. He received the M.S. degree in Engineering Mathematics from Lund University, Lund, Sweden, in 2007. He received the Ph.D. degree in Telecommunications from the Department of Signal Processing at KTH Royal Institute of Technology, Stockholm, Sweden, in 2011. Dr. Björnson was one of the first recipients of the International Postdoc Grant from the Swedish Research Council. This grant funded a joint postdoctoral position from Sept. 2012 to July 2014 at the Alcatel-Lucent Chair on Flexible Radio, Supélec, Paris, France, and the Department of Signal Processing at KTH Royal Institute of Technology, Stockholm, Sweden. From 2014, Dr. Björnson is an Assistant Professor in the tenure-track at the Division of Communication Systems at Linköping University, Linköping, Sweden.

His research interests include multi-antenna cellular communications, massive MIMO techniques, radio resource allocation, green energy efficient systems, and network topology design. He is the first author of the monograph “Optimal Resource Allocation in Coordinated Multi-Cell Systems” published in Foundations and Trends in Communications and Information Theory, 2013. He is also dedicated to reproducible research and has made a large amount of simulation code publically available.

For his work on optimization of multi-cell MIMO communications, he received a Best Paper Award at the 2009 International Conference on Wireless Communications \& Signal Processing (WCSP’09) and a Best Student Paper Award at the 2011 IEEE International Workshop on Computational Advances in Multi-Sensor Adaptive Processing (CAMSAP’11).

\end{IEEEbiographynophoto}
\begin{IEEEbiographynophoto}{M\'erouane Debbah}
 Mérouane Debbah entered the Ecole Normale Supérieure de Cachan (France) in 1996 where he received his M.Sc and Ph.D. degrees respectively. He worked for Motorola Labs (Saclay, France) from 1999-2002 and the Vienna Research Center for Telecommunications (Vienna, Austria) until 2003. He then joined the Mobile Communications department of the Institut Eurecom (Sophia Antipolis, France) as an Assistant Professor until 2007. He is now a Full Professor at Supelec (Gif-sur-Yvette, France), holder of the Alcatel-Lucent Chair on Flexible Radio and a recipient of the ERC grant MORE (Advanced Mathematical Tools for Complex Network Engineering). His research interests are in information theory, signal processing and wireless communications. He is a senior area editor for IEEE Transactions on Signal Processing and an Associate Editor in Chief of the journal Random Matrix: Theory and Applications. Mérouane Debbah is the recipient of the "Mario Boella" award in 2005, the 2007 General Symposium IEEE GLOBECOM best paper award, the Wi-Opt 2009 best paper award, the 2010 Newcom++ best paper award, the WUN CogCom Best Paper 2012 and 2013 Award, the 2014 WCNC best paper award as well as the Valuetools 2007, Valuetools 2008, CrownCom2009 and Valuetools 2012 best student paper awards. In 2011, he received the IEEE Glavieux Prize Award and in 2012, the Qualcomm Innovation Prize Award. He is a WWRF fellow and a member of the academic senate of Paris-Saclay.

\end{IEEEbiographynophoto}
  
      \end{document}